%% file: main.tex
\documentclass[11pt]{article}
\usepackage{amsmath,amsfonts,amsthm,mathtools}
 
\usepackage{amssymb}
\usepackage{bbm}
\usepackage{booktabs}
\usepackage{longtable} 
\usepackage{needspace} 
\usepackage{float} 
\usepackage{placeins} 
\usepackage{array}
\usepackage{enumitem}
\usepackage{algorithm,algpseudocode}
\usepackage{comment}
\usepackage{xcolor}

\usepackage[top=1in, bottom=1in, left=1.25in, right=1in]{geometry}
\usepackage[authoryear]{natbib}
\usepackage{hyperref}
\hypersetup{
  colorlinks=true,
  linkcolor=blue,
  citecolor=blue,
  urlcolor=blue
}

\numberwithin{equation}{section}
\theoremstyle{plain}
\newtheorem{theorem}{Theorem}[section]
\newtheorem{lemma}[theorem]{Lemma}
\newtheorem{proposition}[theorem]{Proposition}
\newtheorem{corollary}[theorem]{Corollary}

\theoremstyle{definition}
\newtheorem{definition}[theorem]{Definition}
\newtheorem{assumption}[theorem]{Assumption}
\newtheorem{remark}[theorem]{Remark}

\theoremstyle{remark}

\newcommand{\EE}{\mathbb{E}} 

\newcommand{\PP}{\mathbb{P}} 
\newcommand{\Var}{\mathrm{Var}}

\newcommand{\R}{\mathbb{R}}
\newcommand{\Ncal}{\mathcal{N}}

\newcommand{\Ecal}{\mathcal{E}}

\newcommand{\Hcal}{\mathcal{H}}

\title{Leave-One-Out Neighborhood Smoothing for Graphons:\\
Berry--Esseen Bounds and Edge-Wise Confidence Intervals}

\author{Behzad Aalipur\thanks{Department of Mathematical Sciences, University of Cincinnati.
Correspondence: \texttt{aalipubd@mail.uc.edu}}
\and
Rachel Kilby\thanks{Department of Mathematical Sciences, University of Cincinnati.
Correspondence: \texttt{kilbyre@mail.uc.edu}}
}

\date{}
\begin{document}
\maketitle

\begin{abstract}
We develop a framework for edge-wise inference in graphon-based random
networks. The proposed leave-one-out neighborhood procedure excludes the
target vertex when selecting neighbors, then averages the selected vertices'
edges to that target. Conditional on the latent probabilities and selected
neighborhood, these edges are independent. This structure yields finite-sample
concentration inequalities and a Berry--Esseen bound for fluctuations about
the selected-neighborhood mean.

Under blockwise graphon regularity, quantitative two-hop identifiability,
and across-block separation, we control the smoothing bias uniformly over
latent positions, including positions near block boundaries. Together with an
observable variance estimator, this gives Gaussian confidence intervals for
individual edge probabilities with asymptotically correct pointwise coverage.
We also establish finite-sample empirical Bernstein coverage for the
selected-neighborhood mean, Gaussian inference under weaker bandwidth
conditions on identical-profile blocks, and inference for target pairs chosen
measurably from the latent positions. Finally, we derive an expected-risk
identity for leave-one-out cross-validation at each fixed neighborhood size.
Simulations assess interval coverage, variance estimation, and estimation
error.
\end{abstract}

\section{Introduction}
\label{sec:intro}

Graphon models specify network connection probabilities as functions of
unobserved latent vertex positions. For an undirected network with adjacency
matrix $A$, the conditional edge-probability matrix $P$ has entries
$P_{ij}=f(\xi_i,\xi_j)$, and distinct edges are independent Bernoulli
variables given the latent positions
\citep{borgs2008convergent,lovasz2012large}. Statistical analysis under
these models has addressed estimation of $P$ and distributional inference
for network statistics and functionals.

As an illustration, consider a scientific collaboration network in which
an edge records whether two researchers collaborated during a specified
period. The observed edge indicator differs from the corresponding connection
probability under the model. Confidence intervals quantify uncertainty
about this probability rather than reporting only a point estimate.

Probability-matrix estimation has been studied using block approximations,
histogram and least-squares estimators, spectral methods, and neighborhood
smoothing \citep{airoldi2013stochastic,chan2014,chatterjee2015matrix,
gao2015rate,xu2018rates}. In particular, the neighborhood smoother of
\citet*{zhang2017estimating} (ZLZ) selects vertices using empirical
two-hop discrepancies and averages edges associated with vertices having
similar connection profiles. ZLZ establish estimation-error guarantees
for piecewise-Lipschitz graphons. Computational lower bounds for related
estimation problems are studied by \citet{Luo2024lowerbound}. These
results do not, by themselves, establish the sampling distribution or
confidence-interval coverage of an individual smoothed entry.

A separate literature studies inference for network statistics and
functionals. Its targets include subgraph counts and network moments
\citep{Bhattacharya_Chatterjee_Janson_2023,ChatterjeeDanBhattacharya2024},
network comparison \citep{Paulin2023GraphonClustering}, and uncertainty
quantification for graphon estimation \citep{su2022uncertainty}.
Other related methods include network resampling and bootstrap procedures
\citep{green2022bootstrapping,zu2025local,gel2017bootstrap,
nowakowicz2025thesis,aalipur2026bootstrap,shao2024parametric},
conformal node-label prediction \citep{Zargar2023conformal}, and methods for
exchangeable arrays and dyadic observations
\citep{davezies2021empirical,matsushita2021jackknife,
graham2021minimax,cattaneo2024uniform}.

Our objective is edge-wise inference from neighborhood-smoothed statistics,
rather than probability-matrix estimation under an aggregate loss. Since the
neighborhood is selected from the observed network, its effect on the
sampling distribution must be accounted for when constructing confidence
intervals.

For a target pair $(i,j)$, conventional neighborhood selection may use
edges incident to vertex $j$ to choose neighbors of $i$, even though
those edges are subsequently averaged to estimate $P_{ij}$. Conditioning
on the selected neighborhood need not preserve independence of the averaged
edges. Our procedure instead selects neighbors using the graph with vertex
$j$ removed and then averages their observed edges to $j$. Conditional on
the latent positions and the target-deleted graph, the selected neighborhood
is fixed while the averaged edges remain independent Bernoulli variables.
This separation permits finite-sample concentration bounds and a
Berry--Esseen approximation centered at the selected-neighborhood mean.

Inference for $P_{ij}$ also requires control of smoothing bias. The
conditional mean of the
leave-one-out average need not equal $P_{ij}$, so a normal approximation
about that mean is insufficient for inference on the target probability.
We establish uniform neighborhood localization and control of the
resulting bias under blockwise graphon regularity, quantitative two-hop
identifiability, and separation across blocks. These geometric bounds
apply without excluding latent positions near block boundaries.
Under the additional condition that edge probabilities are uniformly
bounded away from zero and one, consistency of an observable variance
estimator and the conditional distributional results yield Gaussian
confidence intervals with asymptotically correct pointwise coverage.

Let $h_n$ denote the deterministic number of vertices selected in each
leave-one-out neighborhood, and let $C_0>0$ be a sufficiently large constant.
Our main contributions are threefold.
\begin{enumerate}[leftmargin=2em]
\item \textbf{Conditional distributional bounds.}
The leave-one-out construction provides conditional independence and
an exact conditional risk decomposition. We establish finite-sample
concentration and empirical Bernstein bounds for the selected-neighborhood
mean, together with a Berry--Esseen bound for the centered statistic.
These results do not require the blockwise geometric assumptions used
to control smoothing bias.

\item \textbf{Localization and smoothing-bias control.}
We prove high-probability neighborhood containment and a uniform
smoothing-bias bound over ordered vertex pairs, including those near
block boundaries. The resulting general conditional MSE upper bound
balances at $h_n\asymp n^{2/3}$. This is a balance point of the derived
bound, not a claim of universal optimality.

\item \textbf{Feasible edge-wise inference.}
Under the stated structural and variance assumptions, we establish
a Gaussian limit centered at $P_{ij}$, consistency of an observable
variance estimator, and Gaussian confidence intervals with asymptotically
correct pointwise coverage when $C_0\log n\le h_n=o(n^{2/3})$.
For target pairs whose
smoothing vertex lies in a block with identical connection profiles,
the bias vanishes on the block-containment event, allowing the weaker
requirement $C_0\log n\le h_n=o(n)$. We also treat target pairs chosen
measurably from the latent positions.
\end{enumerate}

The distinction from the neighborhood smoother of
\citet*{zhang2017estimating} is summarized in
Table~\ref{tab:zlz-comparison}. Our use of leave-one-out selection
changes the conditioning argument for the averaged target-column edges;
the edge-wise Gaussian coverage theorem additionally requires the
geometric and nondegeneracy assumptions stated below. Thus the
inferential guarantees do not follow solely from the estimation bounds
of ZLZ, and we do not assert that our regularity class contains theirs.

\begin{table}[H]
\centering
\caption{Comparison of classical neighborhood smoothing and the present
leave-one-out procedure. Each column describes the guarantees established
for the corresponding method.}
\label{tab:zlz-comparison}
\small
\setlength{\tabcolsep}{4pt}
\begin{tabular}{>{\raggedright\arraybackslash}p{0.16\linewidth}>{\raggedright\arraybackslash}p{0.34\linewidth}>{\raggedright\arraybackslash}p{0.37\linewidth}}
\toprule
 & ZLZ neighborhood smoothing & Present leave-one-out smoother \\
\midrule
Selection & Empirical two-hop comparison from the observed network & Target vertex $j$ deleted before selecting neighbors of $i$ \\
Conditional law & Neighborhood selection can involve the averaged target-column edges & Selected target-column edges remain independent given $P$ and the deleted graph \\
Established results & Probability-matrix estimation error under piecewise-Lipschitz regularity & Conditional Berry--Esseen bound; additional localization, feasible variance normalization, and pointwise coverage under stronger structural assumptions \\
\bottomrule
\end{tabular}
\end{table}

As a supplementary result, we derive an exact expected-risk identity
for leave-one-out cross-validation at each fixed neighborhood size.
This is not an oracle guarantee for the data-selected neighborhood size,
and the present confidence-interval theorems use deterministic sizes.
The simulations examine confidence-interval coverage, variance
estimation, and estimation error under deterministic neighborhood sizes.

\section{Model and inferential procedure}
\label{sec:loo}

\subsection{Graphon model and assumptions}
\label{subsec:model-assumptions}
For $n\ge4$, we observe an undirected simple random graph $G=(V,\mathcal E)$,
where $V=[n]:=\{1,\ldots,n\}$ is the vertex set and
$\mathcal E\subseteq\{\{i,j\}:1\le i<j\le n\}$ is the set of observed
edges. Its adjacency matrix $A\in\{0,1\}^{n\times n}$ records the
connections:
\[
A_{ij}=\mathbbm{1}\{\{i,j\}\in\mathcal E\}
\quad (i\ne j),
\qquad A_{ji}=A_{ij},\qquad A_{ii}=0.
\]
The symmetric matrix $P$ contains the unobserved edge probabilities,
with $P_{ij}\in[0,1]$ for $i\ne j$ and $P_{ii}=0$.
Conditional on $P$, the upper-triangular entries
\[
\{A_{ij}:1\le i<j\le n\}
\]
are mutually independent and satisfy
\begin{equation}
\label{eq:model}
A_{ij}\mid P
\sim
\operatorname{Bernoulli}(P_{ij}),
\qquad
1\le i<j\le n.
\end{equation}

Specifically, let $f:[0,1]^2\to[0,1]$ be a symmetric, measurable
graphon. Each vertex $i$ corresponds to a latent position $\xi_i$,
where $\xi_1,\ldots,\xi_n$ are independent and uniformly distributed
on $[0,1]$. The off-diagonal edge probabilities satisfy
\begin{equation}
\label{eq:graphon-probability}
P_{ij}=f(\xi_i,\xi_j),
\qquad i\neq j.
\end{equation}
For $u\in[0,1]$, write $f(u,\cdot)$ for the graphon slice
$v\mapsto f(u,v)$. For $g\in L^2([0,1])$, define
\[
\|g\|_{L^2}
:=
\left(\int_0^1 |g(v)|^2\,dv\right)^{1/2}.
\]
The blockwise regularity assumption controls these slices within smooth
blocks and permits blocks with identical profiles. Continuity across
block boundaries is not assumed.

The graphon is not uniquely identifiable from the distribution of an
observed graph. Indeed, if $\phi:[0,1]\to[0,1]$ preserves Lebesgue
measure, then $f^\phi(u,v):=f(\phi(u),\phi(v))$ generates the same
distribution of random graphs as $f$. Thus the latent coordinates and the
pointwise representation of $f$ cannot be uniquely determined from the
graph distribution \citep{lovasz2012large}.

The following conditions apply to a representative of the graphon with
a fixed partition into smooth or identical-profile blocks. They are not
invariant under arbitrary measure-preserving transformations because
latent-coordinate distances need not be preserved. The assumptions
specify within-block regularity, conditional edge variance,
and two-hop identifiability and separation.

\begin{assumption}[Blockwise graphon regularity]
\label{as:hybrid}
There exists a fixed integer $K\ge1$, independent of $n$, and a fixed
partition
\[
0=b_0<b_1<\cdots<b_K=1,
\qquad
I_m:=(b_m,b_{m+1}),
\quad
m=0,\ldots,K-1.
\]
The partition endpoints do not depend on $n$. Since the latent variables
have a continuous distribution,
\[
\PP\!\left(
\xi_i\in\{b_0,\ldots,b_K\}
\text{ for some }i
\right)
=0.
\]
All subsequent block-membership statements are therefore understood on
this probability-one event. For each $m=0,\ldots,K-1$, exactly one of the
following two conditions holds.

\begin{enumerate}
    \item \textbf{Smooth block.}
    Viewing graphon slices as elements of $L^2([0,1])$, there exist
    constants $c_m>0$ and $L_m>0$ such that
    \[
    c_m|u-u'|
    \le
    \|f(u,\cdot)-f(u',\cdot)\|_{L^2},
    \qquad
    u,u'\in I_m,
    \]
    and
    \[
    |f(u,v)-f(u',v)|
    \le
    L_m|u-u'|,
    \qquad
    u,u'\in I_m,\quad v\in[0,1].
    \]
    The pointwise Lipschitz condition implies the corresponding upper
    $L^2$ bound
    \[
    \|f(u,\cdot)-f(u',\cdot)\|_{L^2}
    \le
    L_m|u-u'|,
    \qquad u,u'\in I_m,
    \]
    so no separate upper $L^2$ constant is needed.

    \item \textbf{Block with identical probability profiles.}
    For all $u,u'\in I_m$ and all $v\in[0,1]$,
    \[
    f(u,v)=f(u',v).
    \]
    Thus every latent position in $I_m$ has the same probability profile
    $f(u,\cdot)$. The equality holds for every $v$, not merely almost everywhere.
\end{enumerate}
The integer $K$ and all block-specific constants appearing above are fixed
and do not depend on $n$.

\end{assumption}

Throughout, generic positive constants $C$ and $c$ may change from line
to line unless indexed.

\begin{assumption}[Edge probabilities bounded away from zero and one]
\label{as:nondeg}
There exists a constant \(\delta \in (0,1/2)\), independent of $n$, such that the graphon satisfies
\begin{equation}
\label{eq:non-deg}
\delta \le f(u,v) \le 1 - \delta \qquad \text{for all } (u,v) \in [0,1]^2.
\end{equation}
For every $i\neq j$, this condition implies
\[
\delta(1-\delta)
\le
P_{ij}(1-P_{ij})
\le
\frac14,
\]
and therefore ensures uniformly positive conditional edge variances.
\end{assumption}

Assumption~\ref{as:nondeg} supplies the uniform variance lower bounds
used in our Berry--Esseen and variance-normalization arguments. It is
not required for the empirical Bernstein inequality. The assumption
excludes sparse graphon sequences whose edge probabilities approach
zero with $n$.

\begin{assumption}[Two-hop identifiability and block separation]
\label{as:nondeg-block}
For
\[
u,u'\in\bigcup_{m=0}^{K-1} I_m,
\]
define the population two-hop discrepancy
\begin{equation}
\label{eq:D2-definition}
D_2(u,u')
:=
\operatorname*{ess\,sup}_{w\in[0,1]}
\left|
\int_0^1
\{f(u,x)-f(u',x)\}f(x,w)\,dx
\right|.
\end{equation}
The following conditions hold.
\begin{enumerate}
    \item \textbf{Within-block quantitative identifiability.}
    For each smooth block $I_m$ there exists $\kappa_m>0$ such that
    \[
    \sup_{w\in I_m}
    \left|
    \int_0^1 \{f(u,x)-f(u',x)\}f(x,w)\,dx
    \right|
    \ge
    \kappa_m\,
    \|f(u,\cdot)-f(u',\cdot)\|_{L^2}
    \]
    for all $u,u'\in I_m$.

    \item \textbf{Across-block two-hop separation.}
    There exists $\eta_2>0$ such that
    \[
    D_2(u,u')\ge \eta_2
    \]
    whenever $u\in I_m$, $u'\in I_{m'}$, and $m\neq m'$.
\end{enumerate}
The constants $\kappa_m$ and $\eta_2$ are fixed and do not depend on $n$.
\end{assumption}

The within-block condition converts two-hop proximity into closeness
of graphon slices, while across-block separation prevents the selection
of vertices in different blocks. Neither property follows from Lipschitz
regularity alone. Appendix~\ref{app:identifiability-notes} gives a
sufficient condition for within-block identifiability based on a
finite-dimensional space of slice differences and explains the role of
uniform lower bounds.

The localization and bias results cover every ordered target pair,
including vertices arbitrarily close to block boundaries. Let
\begin{equation}
\label{eq:all-pairs}
\Ecal_n:=\{(i,j)\in[n]^2:i\ne j\}
\end{equation}
denote the set of such pairs. Since the latent distribution is continuous,
no sampled latent position coincides with a partition endpoint almost
surely.

The assumptions include dense stochastic block models with a fixed
number of blocks of positive proportions, edge probabilities bounded
away from zero and one, and distinct connection-probability profiles
across blocks. The profiles are identical within each block, and for a
symmetric block model distinct profiles imply positive two-hop separation.
For example, the equal-proportion two-block model with connectivity matrix
\[
B=\begin{pmatrix}0.7&0.3\\0.3&0.7\end{pmatrix}
\]
has two-hop connectivity matrix $B^2/2$ with diagonal entries $0.29$
and off-diagonal entries $0.21$. Hence $D_2(u,u')=0.08$ for vertices
in different blocks, while within-block slice differences vanish.
This gives a concrete instance of both conditions in
Assumption~\ref{as:nondeg-block}.

The quantitative two-hop condition is restrictive: it rules out smooth
blocks on which nonzero slice differences become arbitrarily small
relative to their two-hop images. Positive across-block separation also
excludes partitions containing distinct blocks whose two-hop profiles
can approach one another. These requirements are not consequences of
ordinary blockwise Lipschitz regularity. In addition, the variance
assumption excludes graphon sequences with probabilities tending to
zero or one. The nonidentifiable and degenerate stress-test examples
in Appendix~\ref{app:stress-test-simulations} illustrate some of these
limitations; their numerical results are not covered by the general
Gaussian inference theorem.

\Needspace{11\baselineskip}
\subsection{Leave-One-Out Smoothing for Edge-Wise Inference}
\label{subsec:loo}
Fix an ordered pair $(i,j)$ with $i\ne j$. The procedure consists of three steps:
\begin{enumerate}[leftmargin=2em,topsep=3pt,itemsep=1pt]
\item \emph{Separate selection from the target edges.} Exclude vertex $j$
      and all edges incident to it when constructing the neighborhood
      of $i$.
\item \emph{Select a comparison neighborhood.} Use normalized two-hop
      counts in the reduced graph to select the $h_n$ vertices most
      similar to $i$.
\item \emph{Form an inferential statistic.} Average the edges from those
      vertices to $j$ in the original graph and use their empirical
      variance to form a confidence interval under the conditions
      established below.
\end{enumerate}
\paragraph{Neighborhood construction.}
The superscript $(-j)$ denotes quantities constructed after deleting
vertex $j$ and its incident edges. Let
$A^{(-j)}$ be the $(n-1)\times(n-1)$ adjacency matrix obtained by deleting
row and column $j$ from $A$. Define the two-hop matrices
\begin{equation}
\label{eq:twohop-matrices}
M := \frac{1}{n}A^2,
\qquad
M^{(-j)} := \frac{1}{n-1}\big(A^{(-j)}\big)^2.
\end{equation}
The factors $n$ and $n-1$ normalize the full and reduced two-hop
counts by their respective dimensions.
For distinct $i,k\in V\setminus\{j\}$, define the leave-one-out (LOO) distance
\begin{equation}
  \label{eq:d-j}
  d^{(-j)}(i,k)
  :=
  \max_{\ell \in V \setminus \{i,k,j\}} |M^{(-j)}_{i\ell} - M^{(-j)}_{k\ell}|.
\end{equation}
Let
\[
d^{(-j)}_{i,(1)} \le \cdots \le d^{(-j)}_{i,(n-2)}
\]
denote the order statistics of the set
$\{d^{(-j)}(i,k):k\in V\setminus\{i,j\}\}$.
Fix an integer $h_n$ satisfying $1\le h_n\le n-2$. Define the leave-one-out
neighborhood of vertex $i$ by selecting the $h_n$ nearest neighbors under
$d^{(-j)}$. Ties are resolved by increasing vertex label, so the neighborhood has
exactly $h_n$ vertices and is measurable with respect to $A^{(-j)}$.
The rate results are unchanged by other deterministic tie-breaking rules.
Let $\pi$ be the permutation of $V\setminus\{i,j\}$ that orders the pairs
$(d^{(-j)}(i,k),k)$ lexicographically. We define
\begin{equation}
  \label{eq:neighborhood}
 \Ncal_i^{(-j)} = \big\{ \pi(1), \pi(2), \dots, \pi(h_n) \big\}.
\end{equation}

Here and throughout, $h_n$ is the deterministic integer neighborhood size,
not a neighborhood proportion. Define
\begin{equation}
\label{eq:bandwidth-scales}
\rho_n:=\frac{h_n}{n},
\qquad
r_n:=\left(\frac{\log n}{n}\right)^{1/2},
\qquad
b_n(h_n):=\rho_n+r_n=\frac{h_n}{n}+r_n,
\end{equation}
and assume the neighborhood-size condition
\begin{equation}
\label{eq:geometric-regime}
C_0\log n\le h_n=o(n)
\end{equation}
for a sufficiently large constant $C_0>0$.
The lower bound is used in the uniform candidate-count lemma; the upper
bound ensures $h_n/n\to0$.
For deterministic positive sequences, $a_n\asymp b_n$ means their
ratio is bounded above and below by positive constants for all sufficiently
large $n$; $a_n\ll b_n$ means $a_n=o(b_n)$.

\paragraph{One-sided statistic and confidence interval.}
\begin{definition}[One-sided leave-one-out statistic]
\label{def:loo-smoother}
Given $\Ncal_i^{(-j)}$, define
\begin{equation}
\label{eq:estimator}
\tilde P_{ij}\equiv\tilde P_{ij}^{(-j)}(h_n)
:=\frac{1}{h_n}\sum_{k\in\Ncal_i^{(-j)}}A_{kj}.
\end{equation}
The statistic is one-sided because only the first vertex index is
smoothed: the edges incident to $j$ are averaged over the selected
neighbors of $i$. All confidence intervals in this paper use this
one-sided statistic.
\end{definition}
Because $\Ncal_i^{(-j)}$ is measurable with respect to $A^{(-j)}$,
the selected observations $A_{kj}$ are conditionally independent
Bernoulli variables given $(P,A^{(-j)})$. Their conditional mean is
the selected-neighborhood mean, which need not equal $P_{ij}$.

For $h_n\ge2$, the selected binary observations provide an observable
variance estimator. Section~\ref{sec:main_results} uses this estimator
to construct feasible Gaussian confidence intervals for $P_{ij}$ and
establishes their asymptotic pointwise coverage under the stated
conditions. The empirical Bernstein interval instead provides
finite-sample coverage for the selected-neighborhood mean.

\paragraph{Neighborhood-size selection.}
The confidence-interval results use a prespecified deterministic size
$h_n$. For tuning, the same leave-one-out construction permits a
row-specific cross-validation criterion. Fix $i\in[n]$, and let
$\Hcal$ be a finite, nonempty deterministic grid of candidate sizes:
\begin{equation}
\label{eq:tuning-grid}
\Hcal\subset\{1,2,\dots,n-2\}.
\end{equation}
For each $h\in\Hcal$ and $j\neq i$, construct
$\Ncal_i^{(-j)}(h)$ from the reduced matrix $A^{(-j)}$ and set
\begin{equation}
\label{eq:cv-predictor}
\tilde P_{ij}^{(-j)}(h)
:= \frac{1}{h} \sum_{k \in \Ncal_i^{(-j)}(h)} A_{kj},
\qquad j \neq i.
\end{equation}
Define the LOO cross-validation score
\begin{equation}
\label{eq:cv-score}
\widehat R_i(h):=\frac{1}{n-1}\sum_{j\neq i}\bigl(A_{ij}-\tilde P_{ij}^{(-j)}(h)\bigr)^2.
\end{equation}
We select the smallest minimizer over the grid:
\begin{equation}
\label{eq:cv-selector}
\widehat h_i
:=
\min\left(
\arg\min_{h\in\Hcal}\widehat R_i(h)
\right).
\end{equation}
For each validation term, the edge $A_{ij}$ is used only to evaluate
its corresponding predictor: it is excluded from both neighborhood
selection and the average defining $\tilde P_{ij}^{(-j)}(h)$.
An edge held out in one validation term may nevertheless enter
predictors for other terms. Section~\ref{sec:tuning} establishes an expected-risk identity
for fixed $h$; the inference theorems use deterministic $h_n$.

\subsection{Notation and conventions}
\label{subsec:notation}
Tables~\ref{tab:notation-model} and~\ref{tab:notation-estimation}, together with the conventions paragraph below,
define the notation used in the main results and appendices.

\begingroup
\small
\setlength{\LTpre}{0pt}
\setlength{\LTpost}{10pt}
\begin{longtable}{>{\raggedright\arraybackslash}p{0.25\textwidth}>{\raggedright\arraybackslash}p{0.66\textwidth}}
\caption{Model and geometric notation.}\label{tab:notation-model}\\
\toprule
Symbol & Meaning \\
\midrule
\endfirsthead
\multicolumn{2}{c}{\small Table \thetable\ (continued)}\\
\toprule
Symbol & Meaning \\
\midrule
\endhead
$G=(V,\mathcal E)$ & Observed undirected simple graph; $\mathcal E$ is its set of edges. \\
$[n]=V$ & Vertex set $\{1,\ldots,n\}$. \\
$A$ & Observed symmetric adjacency matrix, with $A_{ii}=0$. \\
$P$ & Latent edge-probability matrix; only its off-diagonal entries are statistical parameters. \\
$f$; $\xi_i$ & Graphon and the $\operatorname{Uniform}[0,1]$ latent position of vertex $i$, with $P_{ij}=f(\xi_i,\xi_j)$. \\
$K$; $b_m$; $I_m$ & Number of fixed latent blocks, their endpoints, and $I_m=(b_m,b_{m+1})$ for $m=0,\ldots,K-1$. \\
$\lambda$; $\|\cdot\|_{L^2}$ & Lebesgue measure on $[0,1]$ and the norm on $L^2([0,1],\lambda)$. \\
$D_2(u,u')$ & Population two-hop discrepancy in \eqref{eq:D2-definition}. \\
$\Ecal_n$ & Set of all ordered target pairs $i\ne j$; see \eqref{eq:all-pairs}. \\
$H_i$ & Event that the smoothing vertex $i$ belongs to a block with identical profiles. \\
\bottomrule
\end{longtable}
\endgroup

\begingroup
\small
\setlength{\LTpre}{0pt}
\setlength{\LTpost}{10pt}
\begin{longtable}{>{\raggedright\arraybackslash}p{0.27\textwidth}>{\raggedright\arraybackslash}p{0.64\textwidth}}
\caption{Smoothing, inference, and tuning notation.}\label{tab:notation-estimation}\\
\toprule
Symbol & Meaning \\
\midrule
\endfirsthead
\multicolumn{2}{c}{\small Table \thetable\ (continued)}\\
\toprule
Symbol & Meaning \\
\midrule
\endhead
$A^{(-j)}$; $M$; $M^{(-j)}$ & Reduced adjacency matrix and the full and reduced normalized two-hop matrices in \eqref{eq:twohop-matrices}. \\
$d^{(-j)}(i,k)$; $\Ncal_i^{(-j)}$ & LOO two-hop distance and the resulting fixed-cardinality neighborhood of the smoothing vertex $i$. \\
$h_n$ & Deterministic integer neighborhood size used in the theory; it is a count, not a proportion. \\
$\rho_n=h_n/n$ & Dimensionless LOO localization proportion. \\
$r_n$; $b_n(h_n)$ & Uniform two-hop approximation scale $r_n=(\log n/n)^{1/2}$ and combined bias scale $b_n(h_n)=h_n/n+r_n$. \\
$\bar P_{ij}^{\,\Ncal}$ & Conditional mean of the one-sided estimator, as defined in \eqref{eq:conditional-neighborhood-mean}. \\
$\tilde P_{ij}$; $\widehat P_{ij}$ & One-sided LOO statistic; ancillary symmetrized point estimate defined in Appendix~\ref{app:zlz-comparison}. \\
$U_{ij}^{(-j)}$; $B_{ij}^{(-j)}$ & Centered stochastic fluctuation and conditional smoothing bias in \eqref{eq:error-decomposition}. \\
$V_{ij}$; $s_{ij}^2$; $\widehat V_{ij}^{\mathrm{emp}}$ & Oracle conditional variance, empirical sample variance, and feasible estimator variance. \\
$d_{\mathrm{ZLZ}}$; $\tau_n$; $Q_i^{\mathrm{ZLZ}}(\tau_n)$ & Classical ZLZ distance, dimensionless neighborhood-quantile proportion, and the corresponding quantile cutoff; see Appendix~\ref{app:zlz-comparison}. \\
$\Ncal_i^{\mathrm{ZLZ}}$; $\widetilde P_{ij}^{\mathrm{ZLZ}}$ & Classical ZLZ neighborhood and one-sided smoother; see Appendix~\ref{app:zlz-comparison}. \\
$h$; $\Hcal$; $\bar h_n$ & Generic candidate integer size, finite tuning grid, and deterministic upper bound for a grid satisfying the neighborhood-size condition. \\
$\widehat R_i(h)$; $R_i(h)$; $\widehat h_i$ & Observable CV score and row-specific size selector defined in Subsection~\ref{subsec:loo}; row-averaged conditional MSE defined in Section~\ref{sec:tuning}. \\
\bottomrule
\end{longtable}
\endgroup

\paragraph{General notation and conventions.}
We write $\PP$, $\EE$, and $\Var$ for probability, expectation, and variance,
and $\mathcal L(X\mid\mathcal F)$ for the conditional law given $\mathcal F$.
Let $\mathbb R$ denote the real numbers and $\mathbb R_+=[0,\infty)$.
For sets $S,T$ and an event $E$, $|S|$, $S\setminus T$, and
$\mathbbm{1}\{E\}$ denote cardinality, set difference, and the event
indicator. For a matrix $B$, $B_{i\cdot}$ denotes its $i$th row, and
$\log$ denotes the natural logarithm. We use $O$ and $o$ for deterministic
asymptotic bounds, $\asymp$ for comparability up to fixed positive constants,
and $a_n\ll b_n$ to mean $a_n=o(b_n)$. The symbols $\mathcal O_\PP$ and
$o_\PP$ denote boundedness in probability and convergence to zero in
probability, respectively, after scaling; $\xrightarrow{d}$ and
$\xrightarrow{\PP}$ denote convergence in distribution and in probability.
Let $\Phi$ denote the standard normal distribution function, $z_p=\Phi^{-1}(p)$
its $p$-quantile, and $N(0,1)$ the standard normal law. For
$E,F\subseteq\mathbb R$, $E\oplus F=\{e+f:e\in E,\ f\in F\}$ denotes their
Minkowski sum.

\Needspace{3\baselineskip}
\noindent Throughout, $C,c$ denote generic positive constants, and
admissible indices satisfy the stated distinctness and exclusion conditions,
with nonempty indexing sets in maxima and averages.

\Needspace{17\baselineskip}
\section{Conditional independence, localization, and risk}
\label{sec:bias-geometry}
\subsection{Conditional independence and error decomposition}
The conditional mean of the one-sided estimator is
\begin{equation}
\label{eq:conditional-neighborhood-mean}
\bar P_{ij}^{\,\Ncal}
:=\frac{1}{h_n}\sum_{k\in\Ncal_i^{(-j)}}P_{kj}
=\EE[\tilde P_{ij}\mid \Ncal_i^{(-j)},P].
\end{equation}

The error decomposes into stochastic fluctuation and smoothing bias:
\begin{equation}
\label{eq:error-decomposition}
\tilde P_{ij}-P_{ij}=U_{ij}^{(-j)}+B_{ij}^{(-j)}.
\end{equation}
Here the stochastic component is defined as
\begin{equation}
\label{eq:Uij_def}
    U_{ij}^{(-j)} 
    := \frac{1}{h_n}\sum_{k \in \Ncal_i^{(-j)}} (A_{kj} - P_{kj}),
\end{equation}
and the conditional bias component is defined as
\begin{equation}
\label{eq:Bij_def}
    B_{ij}^{(-j)} 
    := \frac{1}{h_n}\sum_{k \in \Ncal_i^{(-j)}} P_{kj} - P_{ij}.
\end{equation}
\begin{proposition}[Conditional independence]
\label{prop:independence}
Conditional on $\Ncal_i^{(-j)}$ and $P$, the selected target-column
variables are mutually independent and, for every
$k\in\Ncal_i^{(-j)}$,
\[
A_{kj}\mid(\Ncal_i^{(-j)},P)
\sim
\operatorname{Bernoulli}(P_{kj}).
\]
Equivalently, for every realization $S$ of $\Ncal_i^{(-j)}$ with positive
conditional probability given $P$,
\[
\mathcal L\!\left(
(A_{kj})_{k\in S}
\,\middle|\,
\Ncal_i^{(-j)}=S,P
\right)
=
\bigotimes_{k\in S}\operatorname{Bernoulli}(P_{kj}).
\]
Consequently,
$\{A_{kj}-P_{kj}:k\in\Ncal_i^{(-j)}\}$ are conditionally independent,
mean zero, and supported in $[-1,1]$.
\end{proposition}
\subsection{Neighborhood localization and smoothing-bias control}
Bias control requires the selected neighbors to have connection
probabilities close to those of vertex $i$. Uniform two-hop approximation,
across-block separation, and within-block identifiability give the
localization and bias bounds below; their proofs appear in
Appendices~\ref{app:geometry-lemmas} and~\ref{app:geometry-proofs}.

\begin{theorem}[Uniform two-hop approximation and neighborhood containment]
\label{thm:geo-localization}
Suppose Assumptions~\ref{as:hybrid} and~\ref{as:nondeg-block}
hold and the deterministic neighborhood-size condition
\eqref{eq:geometric-regime} is satisfied with sufficiently large
$C_0$. Then there are constants $C,c>0$ and an event
$\Ecal_{\mathrm{master}}$ with probability at least $1-Cn^{-c}$ on
which the following statements hold simultaneously:
\begin{enumerate}
\item For every triple of distinct indices $i,k,j$,
\[
|d^{(-j)}(i,k)-D_2(\xi_i,\xi_k)|\le Cr_n.
\]
\item For every $(i,j)\in\Ecal_n$, all vertices in
$\Ncal_i^{(-j)}$ lie in the same block as $\xi_i$.
\item If $\xi_i$ belongs to a smooth block, every
$k\in\Ncal_i^{(-j)}$ satisfies
\[
|\xi_k-\xi_i|\le Cb_n(h_n),\qquad
\sup_{v\in[0,1]}|f(\xi_k,v)-f(\xi_i,v)|\le Cb_n(h_n).
\]
If the block has identical probability profiles, then
$f(\xi_k,v)=f(\xi_i,v)$ for all $v\in[0,1]$.
\end{enumerate}
\end{theorem}
\begin{theorem}[Uniform smoothing-bias bound]
\label{thm:bias-goodset}
Under the conditions of Theorem~\ref{thm:geo-localization}, on the
same event $\Ecal_{\mathrm{master}}$,
\begin{equation}
\label{eq:uniform-bias-bound}
\sup_{(i,j)\in\Ecal_n}|B_{ij}^{(-j)}|\le Cb_n(h_n).
\end{equation}
If the smoothing vertex belongs to a block with identical
probability profiles, then $B_{ij}^{(-j)}=0$ on this event.
\end{theorem}

\Needspace{7\baselineskip}
The uniform bound controls the squared-bias contribution to risk.
For a Gaussian approximation centered at $P_{ij}$, the stronger
condition $\sqrt{h_n}\,b_n(h_n)\to0$ makes the bias negligible relative
to the standard deviation; Section~\ref{sec:main_results} uses this
condition for inference.

\subsection{Conditional risk and neighborhood-size tradeoff}
Proposition~\ref{prop:independence} yields the conditional
variance--bias identity.
\begin{theorem}[Conditional mean squared error]
\label{thm:loo-mse}
Conditional on the latent probability matrix $P$ and the LOO neighborhood
$\Ncal_i^{(-j)}$ of size $h_n$, the MSE of the LOO estimator satisfies
\begin{equation}
\label{eq:risk-bound}
\begin{aligned}
\EE\left[
    (\tilde{P}_{ij}-P_{ij})^2
    \,\middle|\,
    \Ncal_i^{(-j)},P
\right]
&=
\frac{1}{h_n^2}
\sum_{k\in\Ncal_i^{(-j)}}P_{kj}(1-P_{kj})
+
\bigl\{B_{ij}^{(-j)}\bigr\}^2
\\
&\le
\frac{1}{4h_n}
+
\max_{k\in\Ncal_i^{(-j)}}
(P_{kj}-P_{ij})^2.
\end{aligned}
\end{equation}
\end{theorem}

\begin{corollary}[Conditional risk bound for a given neighborhood size]
\label{cor:bandwidth-risk}
Under the assumptions of Theorem~\ref{thm:bias-goodset}, for every
$(i,j)\in\Ecal_n$,
\begin{equation}
\label{eq:bandwidth-risk-bound}
\EE\!\left[
(\tilde P_{ij}-P_{ij})^2
\,\middle|\,
\Ncal_i^{(-j)},P
\right]
\le
\frac{1}{4h_n}+C b_n(h_n)^2
\end{equation}
on the high-probability event established in
Theorem~\ref{thm:bias-goodset}, where $b_n(h_n)$ is defined in
\eqref{eq:bandwidth-scales}.

Consequently, the variance and squared-bias terms in
\eqref{eq:bandwidth-risk-bound} balance at
\begin{equation}
\label{eq:balance-bandwidth}
h_n\asymp n^{2/3}.
\end{equation}
Indeed, at the scale \eqref{eq:balance-bandwidth},
\[
\frac{h_n}{n}\asymp n^{-1/3},
\qquad
r_n=\sqrt{\frac{\log n}{n}}=o(n^{-1/3}),
\]
so $b_n(h_n)\asymp h_n/n$. Hence the terms involving $r_n$ are lower
order at this choice of $h_n$, and the right-hand side of
\eqref{eq:bandwidth-risk-bound} is $\mathcal O(n^{-2/3})$.
This statement concerns the balance point of the derived risk upper bound;
we do not claim a matching minimax lower bound for the present graphon
class.
\end{corollary}

The entrywise risk criterion in Corollary~\ref{cor:bandwidth-risk}
differs from the rowwise criterion studied by ZLZ. For an $n\times n$
matrix $B$, define its $(2,\infty)$ norm and the normalized distance
between two matrices $B,C\in\mathbb R^{n\times n}$ by
\[
\|B\|_{2,\infty}
:=\max_{1\le i\le n}\left(\sum_{j=1}^n B_{ij}^2\right)^{1/2},
\qquad
d_{2,\infty}(B,C)
:=\frac{1}{\sqrt n}\|B-C\|_{2,\infty}.
\]
Theorem~1 and equation~(9) of \citet{zhang2017estimating}
control the squared normalized $(2,\infty)$
loss of their probability-matrix estimator at order
$(\log n/n)^{1/2}$. This corresponds to a maximum normalized row-$L^2$
scale of order $(\log n/n)^{1/4}$. That criterion is not identical to the
entrywise error criterion considered here. Our entrywise bias control
relies on quantitative two-hop identifiability, a lower bound on
within-block slice distances, and pointwise Lipschitz regularity. Together,
these assumptions convert two-hop discrepancy into control of latent
positions and connection probabilities.

\section{Distributional bounds and edge-wise inference}
\label{sec:main_results}
We distinguish finite-sample inference on the selected-neighborhood
mean from asymptotic inference on the edge probability $P_{ij}$.
Both constructions use the one-sided statistic $\tilde P_{ij}$, but
only the latter requires asymptotically negligible smoothing bias.

\subsection{Finite-sample confidence intervals}
\label{subsec:ci_empirical_bernstein}
Throughout this subsection, assume $2\le h_n\le n-2$.

Fix $\alpha\in(0,1)$. To construct an empirical Bernstein interval
for the selected-neighborhood mean, define the sample variance of the
selected binary observations by
\begin{equation}
\label{eq:empirical-sample-variance}
s_{ij}^2
:=
\frac{1}{h_n-1}\sum_{k\in\Ncal_i^{(-j)}}(A_{kj}-\tilde P_{ij})^2
=
\frac{h_n}{h_n-1}\tilde P_{ij}(1-\tilde P_{ij}).
\end{equation}
The identity in \eqref{eq:empirical-sample-variance} holds for
binary observations without an identical-distribution assumption.
In general, $s_{ij}^2$ differs from the average conditional variance
\[
\frac1{h_n}\sum_{k\in\Ncal_i^{(-j)}}P_{kj}(1-P_{kj}).
\]
We use this observable sample variance in the empirical Bernstein
half-width
\begin{equation}
\label{eq:bernstein-width}
    w_{ij}^{(B)}(\alpha) 
    := 
    \sqrt{\frac{2 s_{ij}^2 \log(4/\alpha)}{h_n}} 
    + 
    \frac{7 \log(4/\alpha)}{3(h_n - 1)}.
\end{equation}
The empirical Bernstein inequality applies to independent, bounded
variables that need not have identical distributions
\citep[Theorem~11]{maurer2009empirical}. Applying it to the variables
and their complements yields the following two-sided bound.

\begin{lemma}[Two-sided empirical Bernstein inequality for independent bounded variables]
\label{lem:empirical-bernstein}
Let $m\ge2$ and let $X_1,\dots,X_m$ be independent random variables taking values in $[0,1]$, not necessarily identically distributed. Define
\[
\bar{X} := \frac{1}{m}\sum_{k=1}^m X_k,
\qquad
\mu := \frac{1}{m}\sum_{k=1}^m \EE[X_k],
\]
and the empirical variance
\[
s^2 := \frac{1}{m-1}\sum_{k=1}^m (X_k - \bar{X})^2.
\]
Then for any $\alpha \in (0,1)$,
\begin{equation}
\label{eq:empirical-bernstein-two-sided}
\PP\!\left(
|\bar{X} - \mu|
\le
\sqrt{\frac{2 s^2 \log(4/\alpha)}{m}}
+
\frac{7 \log(4/\alpha)}{3(m-1)}
\right)
\ge
1 - \alpha.
\end{equation}
\end{lemma}

To express the bias adjustment, we use the Minkowski sum
$E\oplus F:=\{e+f:e\in E,\ f\in F\}$ for subsets $E,F\subset\mathbb R$.

\begin{theorem}[Finite-sample coverage with a bias bound]
\label{thm:finite-ci}
Suppose Assumptions~\ref{as:hybrid} and~\ref{as:nondeg-block}
hold, and let $h_n$ satisfy \eqref{eq:geometric-regime}.
For each fixed deterministic pair $i\ne j$ and $\alpha\in(0,1)$,
\begin{equation}
\label{eq:bernstein-ci}
\mathrm{CI}^{(B)}_{ij}(\alpha)
:=
[\tilde P_{ij}-w_{ij}^{(B)}(\alpha),\,
 \tilde P_{ij}+w_{ij}^{(B)}(\alpha)]\cap[0,1]
\end{equation}
covers the selected-neighborhood mean
$\bar P_{ij}^{\,\Ncal}$ with conditional probability at least
$1-\alpha$, given the latent positions and the leave-one-out
neighborhood. Moreover, there exist constants $C,c>0$ such that,
for all sufficiently large $n$,
\begin{equation}
\label{eq:bernstein-coverage}
\PP\!\left(
P_{ij}\in
\bigl(\mathrm{CI}^{(B)}_{ij}(\alpha)
\oplus[-Cb_n(h_n),Cb_n(h_n)]\bigr)\cap[0,1]
\right)\ge1-\alpha-Cn^{-c}.
\end{equation}
\end{theorem}
\begin{remark}[The unknown bias constant]
\label{rem:bias-constant}
The half-width $w_{ij}^{(B)}(\alpha)$ is observable, but the bias-margin
constant $C$ depends on Assumptions~\ref{as:hybrid} and
\ref{as:nondeg-block}. Unless these constants are known or bounded, the
bias-adjusted interval in Theorem~\ref{thm:finite-ci} cannot be constructed
from the observed network alone. The simulations therefore report
empirical Bernstein coverage for the selected-neighborhood mean. For
feasible inference on $P_{ij}$, they use the Gaussian interval with a
neighborhood size satisfying the stated asymptotic condition.
\end{remark}

\begin{remark}[Complementary roles of Bernstein and Gaussian inference]
\label{rem:bernstein-vs-gaussian}
The empirical Bernstein interval provides finite-sample coverage without
a normal approximation. Its truncation to $[0,1]$ can, however, yield
the entire parameter space. Under the stated assumptions, the Gaussian
interval has asymptotic coverage for $P_{ij}$ and width of order
$h_n^{-1/2}$; its approximation error includes both the Berry--Esseen
term and the smoothing-bias contribution.
\end{remark}

\subsection{Normal approximation and Gaussian confidence intervals}
\label{subsec:ci_normal}
Let $\Phi$ denote the standard normal distribution function and
$N(0,1)$ the standard normal law. We use $z_p=\Phi^{-1}(p)$ for the
$p$-quantile of this distribution.
Define the conditional variance
\begin{equation}
\label{eq:oracle-variance}
V_{ij}
:=
\Var\!\big(U_{ij}^{(-j)} \mid \Ncal_i^{(-j)}, P\big)
=
\frac{1}{h_n^2}\sum_{k\in\Ncal_i^{(-j)}} P_{kj}(1-P_{kj}).
\end{equation}
Under Assumption~\ref{as:nondeg}, $V_{ij}$ is uniformly of order
$h_n^{-1}$. For the Berry--Esseen bound, define
\begin{equation}
\label{eq:standardization-identity}
S_{ij,n}:=\sum_{k\in\Ncal_i^{(-j)}}(A_{kj}-P_{kj}),
\qquad
\sigma_{ij,n}^2:=\sum_{k\in\Ncal_i^{(-j)}}\Var(A_{kj}-P_{kj}\mid\Ncal_i^{(-j)},P),
\qquad
\frac{U_{ij}^{(-j)}}{\sqrt{V_{ij}}}=\frac{S_{ij,n}}{\sigma_{ij,n}}.
\end{equation}
By Proposition~\ref{prop:independence}, the centered summands are
conditionally independent, though not necessarily identically
distributed. The Berry--Esseen inequality for independent triangular
arrays \citep{Petrov1975}, together with the variance lower bound,
gives the following approximation without geometric assumptions.
\begin{lemma}[Lower variance bound]
\label{lem:var-lower}
Under Assumption~\ref{as:nondeg}, the conditional variance in \eqref{eq:oracle-variance} satisfies
\begin{equation}
\label{eq:variance-bounds}
\frac{\delta(1-\delta)}{h_n}\le V_{ij}\le\frac{1}{4h_n}.
\end{equation}
Thus $V_{ij}\asymp h_n^{-1}$ uniformly in $(i,j)$.
\end{lemma}
\begin{theorem}[Berry--Esseen bound for fixed target pairs]
\label{thm:be}
Suppose Assumption~\ref{as:nondeg} holds and $h_n\to\infty$. Then there
exists a constant $C>0$, depending only on $\delta$ and the universal
Berry--Esseen constant, such that, uniformly over all distinct pairs
$i\neq j$, conditional on $(\Ncal_i^{(-j)},P)$,
\begin{equation}
\label{eq:stochastic-be-bound}
\sup_{x\in\R}
\left|
\PP\!\left(
\frac{U_{ij}^{(-j)}}{\sqrt{V_{ij}}}\le x \,\middle|\, \Ncal_i^{(-j)},P
\right)
-\Phi(x)
\right|
\le
C h_n^{-1/2}.
\end{equation}
Because the conditional bound holds almost surely with a deterministic
constant uniform in the realized neighborhood and probability matrix,
integrating the conditional distribution function over
$(\Ncal_i^{(-j)},P)$ yields the same $Ch_n^{-1/2}$ Kolmogorov bound for
the corresponding unconditional distribution.
\end{theorem}
\begin{remark}[Pointwise versus simultaneous inference]
The word ``uniform'' in Theorem~\ref{thm:be} means that the same pointwise
Kolmogorov bound holds for every deterministic ordered pair. It does not bound
the joint maximum over the $n(n-1)$ standardized errors and therefore does
not supply simultaneous confidence coverage for a row or the full matrix.
\end{remark}

\begin{corollary}[Berry--Esseen bound for the full LOO estimation error]
\label{cor:full-be}
Suppose Assumptions~\ref{as:hybrid}, \ref{as:nondeg}, and
\ref{as:nondeg-block} hold, and let $h_n$ satisfy
\eqref{eq:geometric-regime}. There exist constants $C,c>0$ such
that for every deterministic pair $i\ne j$,
\begin{equation}
\label{eq:full-be-bound}
\sup_{x\in\mathbb R}
\left|\PP\!\left(
\frac{\tilde P_{ij}-P_{ij}}{\sqrt{V_{ij}}}\le x
\right)-\Phi(x)\right|
\le C\bigl\{h_n^{-1/2}+\sqrt{h_n}\,b_n(h_n)+n^{-c}\bigr\}.
\end{equation}
\end{corollary}

To center the normal approximation at $P_{ij}$ rather than the
selected-neighborhood mean, the bias must be negligible compared with
the standard deviation. Theorem~\ref{thm:bias-goodset} and
\eqref{eq:variance-bounds} show that it suffices to require
\begin{equation}
\label{eq:clt-neighborhood-size-condition}
\sqrt{h_n}\,b_n(h_n)
=
\frac{h_n^{3/2}}{n}
+
\sqrt{\frac{h_n\log n}{n}}
\longrightarrow0.
\end{equation}

\begin{theorem}[Pointwise asymptotic normality without interior trimming]
\label{thm:full-clt}
Suppose Assumptions~\ref{as:hybrid}, \ref{as:nondeg}, and
\ref{as:nondeg-block} hold. For deterministic ordered indices
$i\ne j$, let the neighborhood sizes satisfy
\eqref{eq:geometric-regime} and
\eqref{eq:clt-neighborhood-size-condition}. Under
\eqref{eq:geometric-regime}, the latter condition is equivalent to
\begin{equation}
\label{eq:neighborhood-size-regime}
h_n=o(n^{2/3}).
\end{equation}
Then, without conditioning on a trimmed latent region,
\[
\frac{\tilde P_{ij}-P_{ij}}{\sqrt{V_{ij}}}
\xrightarrow{d}N(0,1).
\]
The same conclusion holds uniformly in the choice of deterministic
ordered pair in the sense of the Kolmogorov bound
\eqref{eq:full-be-bound}; no simultaneous-inference claim is made.
\end{theorem}
\begin{remark}[Scope of the asymptotic normality result]
Theorem~\ref{thm:full-clt} establishes asymptotic normality using the oracle conditional variance $V_{ij}$. Feasible inference also requires ratio consistency of the variance estimator. The theorem assumes a deterministic neighborhood size. Its rate conditions alone do not establish validity after data-dependent neighborhood-size selection, which would require a separate post-selection analysis.
\end{remark}
\begin{remark}[Estimation versus inference bandwidths]
\label{rem:estimation-inference-neighborhood-sizes}
The conditional risk upper bound balances at $h_n\asymp n^{2/3}$,
whereas at this scale
\[
\sqrt{h_n}\,\frac{h_n}{n}
\asymp 1.
\]
Thus the smoothing bias need not vanish relative to the stochastic
fluctuation. Centered Gaussian inference for $P_{ij}$ instead uses a
deterministic sequence satisfying both \eqref{eq:geometric-regime}
and \eqref{eq:neighborhood-size-regime}. For example,
$h_n=\lceil n^\gamma\rceil$ with fixed $\gamma\in(0,2/3)$
satisfies these conditions for all sufficiently large $n$.
\end{remark}

The variance in \eqref{eq:oracle-variance} is not observable. We use the
empirical sample variance from \eqref{eq:empirical-sample-variance} to
define the feasible estimator
\begin{equation}
\label{eq:feasible-variance}
\widehat V_{ij}^{\mathrm{emp}}:=\frac{s_{ij}^2}{h_n}.
\end{equation}

\begin{corollary}[Feasible Gaussian normalization and confidence interval]
\label{cor:feasible-normalization}
Under the conditions of Theorem~\ref{thm:full-clt}, let
$\widehat V_{ij}^{\mathrm{emp}}$ be defined by
\eqref{eq:feasible-variance}. Define the studentized statistic to
be zero when $\widehat V_{ij}^{\mathrm{emp}}=0$. Then
\[
\frac{\tilde P_{ij}-P_{ij}}
{\sqrt{\widehat V_{ij}^{\mathrm{emp}}}}
\xrightarrow{d}N(0,1),
\]
where the fraction is understood with this convention.
For fixed $\alpha\in(0,1)$, the interval
\begin{equation}
\label{eq:gaussian-ci}
\mathrm{CI}_{ij}^{(N)}(\alpha)
:=\left[\tilde P_{ij}\pm z_{1-\alpha/2}
\sqrt{\widehat V_{ij}^{\mathrm{emp}}}\right]\cap[0,1]
\end{equation}
satisfies
\begin{equation}
\label{eq:gaussian-coverage}
\PP\bigl(P_{ij}\in\mathrm{CI}_{ij}^{(N)}(\alpha)\bigr)
\longrightarrow1-\alpha.
\end{equation}
\end{corollary}

\begin{theorem}[Gaussian inference on blocks with identical profiles]
\label{thm:identical-block-clt}
Suppose Assumptions~\ref{as:hybrid}, \ref{as:nondeg}, and
\ref{as:nondeg-block} hold. Let $I_{\mathrm{const}}$ denote the
union of blocks with identical connection profiles, and suppose
$\pi_{\mathrm{const}}:=\lambda(I_{\mathrm{const}})>0$.
Fix $\alpha\in(0,1)$. For a fixed smoothing index $i$, define
\[
H_i:=\{\xi_i\in I_{\mathrm{const}}\}.
\]
If the deterministic size satisfies \eqref{eq:geometric-regime}, then
for every fixed $i\ne j$,
\begin{equation}
\label{eq:constant-block-be}
\sup_{x\in\mathbb R}
\left|\PP\!\left(
\frac{\tilde P_{ij}-P_{ij}}{\sqrt{V_{ij}}}\le x
\,\middle|\,H_i\right)-\Phi(x)\right|
\le C\left(h_n^{-1/2}+\pi_{\mathrm{const}}^{-1}n^{-c}\right).
\end{equation}
Moreover, with the zero-variance convention of
Corollary~\ref{cor:feasible-normalization},
\[
\frac{\tilde P_{ij}-P_{ij}}{\sqrt{\widehat V_{ij}^{\mathrm{emp}}}}
\xrightarrow{d}N(0,1)
\quad\text{under }\PP(\,\cdot\mid H_i),
\]
and the feasible Gaussian confidence interval has conditional
coverage tending to $1-\alpha$. For a dense stochastic block model
satisfying the assumptions, $\pi_{\mathrm{const}}=1$, so the
conclusions are unconditional. In particular, the smooth-block
condition $h_n=o(n^{2/3})$ is not required.
\end{theorem}

\begin{theorem}[Inference for latent-selected target pairs]
\label{thm:latent-selected-pair}
Suppose the assumptions of Theorem~\ref{thm:full-clt} hold.
Fix $\alpha\in(0,1)$.
Let $(I_n,J_n)=\psi_n(\xi_1,\ldots,\xi_n)$ be any measurable
pair of distinct indices. Then
\begin{equation}
\label{eq:latent-target-be}
\sup_{x\in\mathbb R}
\left|\PP\!\left(
\frac{\tilde P_{I_nJ_n}-P_{I_nJ_n}}
{\sqrt{V_{I_nJ_n}}}\le x\right)-\Phi(x)\right|
\le C\bigl\{h_n^{-1/2}+\sqrt{h_n}\,b_n(h_n)+n^{-c}\bigr\}.
\end{equation}
The corresponding feasible standardized error converges to the
standard normal distribution, and its pointwise Gaussian confidence
interval has coverage tending to $1-\alpha$.
The target pair may depend on the latent variables but not on
observed edges. The result does not justify inference for indices
selected from the adjacency matrix.
\end{theorem}
\section{Expected-risk identity for neighborhood-size tuning}
\label{sec:tuning}

For each fixed $h\in\Hcal$, conditional independence of the held-out
edge $A_{ij}$ and its predictor $\tilde P_{ij}^{(-j)}(h)$, given $P$,
eliminates the mixed term in the expected squared validation error:
\begin{equation}
\label{eq:true-risk}
\begin{aligned}
&\EE\!\left[
\bigl(A_{ij}-\tilde P_{ij}^{(-j)}(h)\bigr)^2
\,\middle|\,
\Ncal_i^{(-j)}(h),P
\right]
\\
&\qquad=
P_{ij}(1-P_{ij})
+
\EE\!\left[
\bigl(\tilde P_{ij}^{(-j)}(h)-P_{ij}\bigr)^2
\,\middle|\,
\Ncal_i^{(-j)}(h),P
\right].
\end{aligned}
\end{equation}
Define the row-averaged conditional mean squared error as
\begin{equation}
\label{eq:row-risk}
R_i(h)
:=
\frac{1}{n-1}\sum_{j\neq i}
\EE\!\left[
\bigl(\tilde P_{ij}^{(-j)}(h)-P_{ij}\bigr)^2
\,\middle|\,
P
\right].
\end{equation}
Averaging \eqref{eq:true-risk} over $j\neq i$ and then over the randomness
of the LOO neighborhoods gives the exact identity
\begin{equation}
\label{eq:cv-risk-identity}
\EE[\widehat R_i(h)\mid P]
=
\frac{1}{n-1}\sum_{j\neq i}P_{ij}(1-P_{ij})
+
R_i(h).
\end{equation}
The first term is independent of $h$. Hence the minimizers of the
expected CV score and $R_i(h)$ coincide over $\Hcal$.

For estimation-oriented tuning, candidate sizes may be taken near
the $n^{2/3}$ balance point of the derived risk upper bound.
For the general smooth-block Gaussian theorem, a deterministic grid
whose candidates satisfy the undersmoothing condition can be specified by
\[
\Hcal
\subset
\left\{
h:
C_0\log n\le h\le \bar h_n
\right\},
\qquad
\bar h_n=o(n^{2/3}).
\]

\begin{remark}[Deterministic neighborhood size in the inferential theory]
The risk and inference theorems use a deterministic common neighborhood
size $h_n$, not the row-specific selector $\widehat h_i$. Restricting
the tuning grid to admissible sizes does not by itself establish
coverage after selection. Such a result would require additional control
of the selected estimator and the dependence among the CV scores. For
inference under the present theory, a prespecified deterministic
neighborhood size is sufficient.
\end{remark}

\section{Simulation study}
\label{sec:sim}

We evaluate the feasible Gaussian intervals for $P_{ij}$, the empirical
Bernstein intervals for the selected-neighborhood mean, and the estimation
error and bandwidth sensitivity of the smoothing statistic. All experiments
use deterministic neighborhood sizes: $h_n\asymp n^{2/3}$ for estimation
comparisons and an undersmoothed size for Gaussian inference. The
cross-validation selector is not evaluated. Appendix~\ref{app:stress-test-simulations}
contains experiments on four additional graphons, some outside the scope
of the assumptions.

The simulations use \texttt{Python} and \texttt{NumPy} with fixed random
seeds and comprise $2{,}180$ independent graph replications: $1{,}290$
for the three principal graphons, $430$ for the oscillatory stress test,
$360$ for the three additional stress tests, and $100$ for the separate
bandwidth-sensitivity experiment. Distinct neighborhood sizes and target
pairs evaluated on the same graph do not count as additional graph
replications. The code, raw target-level records, aggregate summaries,
and environment details accompany the manuscript.

\subsection{Graphon models and Monte Carlo design}
\label{subsec:sim-design}

We use three graphons satisfying the assumptions of the main theory:
\begin{align}
f_{\mathrm{kink}}(u,v)
&=
0.2+0.25\{\phi(u)+\phi(v)\},
&
\phi(u)
&=
u+0.6|u-1/2|,
\label{eq:sim-kink}\\
f_{\mathrm{lin}}(u,v)
&=
0.2+0.3(u+v).
\label{eq:sim-linear}
\end{align}
The third model is a two-block stochastic block model with equal block
proportions, within-block probability $0.7$, and between-block probability
$0.3$.

The first two models are globally bi-Lipschitz in the latent coordinate and
are uniformly bounded away from zero and one. For $f_{\mathrm{kink}}$,
\[
\phi(u)
=
\begin{cases}
0.3+0.4u, & 0\le u\le1/2,\\
1.6u-0.3, & 1/2\le u\le1.
\end{cases}
\]
Consequently,
\[
f_{\mathrm{kink}}([0,1]^2)=[0.35,0.85],
\qquad
f_{\mathrm{lin}}([0,1]^2)=[0.2,0.8].
\]
The slopes $0.4$ and $1.6$ make $\phi$ bi-Lipschitz. The kink at
$1/2$ gives an admissible Lipschitz graphon that is not everywhere
differentiable.

For both additive models, quantitative two-hop identifiability can be checked
directly. More generally, if
\[
f(u,v)=a+b\{\psi(u)+\psi(v)\},
\]
then
\[
f(u,x)-f(u',x)=b\{\psi(u)-\psi(u')\}
\]
is constant in $x$, and therefore
\[
\left|
\int_0^1
\{f(u,x)-f(u',x)\}f(x,w)\,dx
\right|
=
|b|\,|\psi(u)-\psi(u')|
\int_0^1 f(x,w)\,dx.
\]
Moreover,
\[
\|f(u,\cdot)-f(u',\cdot)\|_{L^2}
=
|b|\,|\psi(u)-\psi(u')|.
\]
Since both additive simulation graphons are bounded below by a positive
constant $\delta$, Assumption~\ref{as:nondeg-block}(1) holds with a constant
$\kappa_m\ge\delta$.

In the equal two-block model, all latent positions within a given block have the same connection-probability profile. Writing
\[
B=
\begin{pmatrix}
0.7&0.3\\
0.3&0.7
\end{pmatrix},
\qquad
\frac12B^2
=
\begin{pmatrix}
0.29&0.21\\
0.21&0.29
\end{pmatrix},
\]
we obtain $D_2(u,u')=0.08$ whenever $u$ and $u'$ lie in different
blocks. Thus the identical-profile condition within blocks and the across-block separation
condition hold.

We consider
\[
n\in\{200,500,1000,2000\}.
\]
The corresponding numbers of independently generated graph replications are
$150,120,100,$ and $60$. Within each graph, we evaluate four target
coordinate pairs:
\[
(0.25,0.75),\quad
(0.50,0.80),\quad
(0.75,0.25),\quad
(0.50,0.20).
\]
For each coordinate pair $(u,v)$, the implementation selects the sampled
vertex nearest to $u$ and then the sampled vertex nearest to $v$ among the
remaining vertices. Ties are broken by original vertex label. This convention
ensures that the resulting target indices satisfy $i\ne j$. The corresponding
numbers of target-edge evaluations per model are $600,480,400,$ and $240$.

Unless stated otherwise, we average the four target-level quantities
within each graph and then average across independent graph replications.
Root mean squared conditional bias is the square root of the resulting
mean squared bias. Because target evaluations within a graph are
dependent, Monte Carlo standard errors use graph-level summaries:
for coverage, we compute the standard deviation of the four-target
coverage averages across graphs and divide by the square root of the
number of graph replications. For the LOO Gaussian coverage estimates,
these standard errors range from $0.008$ to $0.013$ in the models
satisfying the theorem assumptions.

For the estimation-rate experiment, let $\operatorname{round}(x)$ denote
rounding to the nearest integer and use
\begin{equation}
\label{eq:simulation-balance-bandwidth}
h_n^{\mathrm{bal}}
=
\operatorname{round}(n^{2/3}).
\end{equation}
For $n\in\{200,500,1000,2000\}$, this gives
\[
h_n^{\mathrm{bal}}\in\{34,63,100,159\}.
\]
We compare the one-sided LOO estimator with the fixed-cardinality ZLZ version
defined in Appendix~\ref{app:zlz-comparison}, using $h_n$ neighbors for both
procedures. This equal-cardinality implementation isolates the effect of the
LOO construction from differences in neighborhood size.

For the inference experiment, we use the deterministic polynomial
sequence satisfying the neighborhood-size condition
\begin{equation}
\label{eq:simulation-inference-bandwidth}
h_n^{\mathrm{inf}}
=
\operatorname{round}(n^{3/5}).
\end{equation}
This sequence satisfies $\log n\ll n^{3/5}=o(n^{2/3})$; the exponent
$3/5$ is illustrative rather than optimal. For the four sample sizes, it gives
\[
h_n^{\mathrm{inf}}\in\{24,42,63,96\}.
\]
We use the feasible variance estimator in \eqref{eq:feasible-variance} and the
$95\%$ specialization of the Gaussian interval in \eqref{eq:gaussian-ci}.
For descriptive comparison only, we construct the analogous Gaussian interval
for the fixed-cardinality ZLZ smoother using its within-neighborhood sample
variance. As discussed in Appendix~\ref{app:zlz-comparison}, this descriptive
interval does not inherit the conditional Berry--Esseen guarantee established
here for the LOO estimator.

For empirical Bernstein inference, the unknown bias constant multiplying
$b_n(h_n)$ prevents a completely data-driven implementation of the
finite-sample interval for $P_{ij}$. We therefore examine the empirical
Bernstein interval for the selected-neighborhood mean
$\bar P_{ij}^{\,\Ncal}$ defined in
\eqref{eq:conditional-neighborhood-mean}.
We report its estimated coverage and truncated width
\[
W_{ij}^{(B)}
:=
\min\{1,\tilde P_{ij}+w_{ij}^{(B)}\}
-
\max\{0,\tilde P_{ij}-w_{ij}^{(B)}\}.
\]
We also report the fraction of empirical Bernstein intervals equal to
$[0,1]$, for which coverage provides no discrimination within the parameter
space. Gaussian interval widths are also computed after intersection with
$[0,1]$; this operation preserves coverage for $P_{ij}\in[0,1]$.

\subsection{Gaussian coverage and variance estimation}

Table~\ref{tab:gaussian-inference} reports coverage and variance
normalization for the feasible Gaussian interval at
$h_n=\operatorname{round}(n^{3/5})$. The ZLZ comparison uses the
fixed-cardinality construction specified above and has no corresponding
coverage theorem.

\input{generated_results/regenerated_20260831/tables/gaussian.tex}

Mean empirical-to-oracle variance ratios range from $0.998$ to
$1.010$, consistent with Lemma~\ref{lem:emp-var-consistency}; these
averages do not show the dispersion of individual ratios. Estimated LOO
Gaussian coverage ranges from $0.929$ to $0.975$, compared with the
nominal $0.95$ level. Average LOO Gaussian widths decrease from
approximately $0.37$--$0.40$ at $n=200$ to $0.18$--$0.20$ at
$n=2000$.

Table~\ref{tab:gaussian-coverage-diagnostics} reports signed coverage
deviations from $0.95$, graph-replication-level Monte Carlo standard
errors, and mean absolute conditional bias standardized by the oracle
standard deviation. The standardized bias averages
$|B_{ij}^{(-j)}|/\sqrt{V_{ij}}$ over target pairs and graph
replications; it is nearly zero for the equal two-block SBM, consistent
with identical profiles and block containment. The signed coverage
deviations are generally of the same order as their graph-level Monte
Carlo standard errors, and the largest displayed deviation is $0.025$.
A descriptive nominal $95\%$ Monte Carlo interval for a reported
coverage probability is the estimate plus or minus $1.96$ times its
replication-level standard error. These diagnostics describe Monte Carlo
uncertainty; they do not establish asymptotic coverage. That result follows
from Corollary~\ref{cor:feasible-normalization} for deterministic pairs
and Theorem~\ref{thm:latent-selected-pair} for the latent-selected targets
used here.

\input{generated_results/regenerated_20260831/tables/coverage_diagnostics_v46.tex}

\FloatBarrier
\subsection{Empirical Bernstein coverage and interval width}
Table~\ref{tab:bernstein-inference} reports estimated empirical Bernstein
coverage for the conditional neighborhood mean $\bar P_{ij}^{\,\Ncal}$.
The empirical Bernstein bound is fully observable for this target.
An interval for $P_{ij}$ would additionally require the generally
unknown bias margin $Cb_n(h_n)$.

\input{generated_results/regenerated_20260831/tables/bernstein.tex}

At $n=200$, the average truncated Bernstein widths are $0.997$, $0.997$,
and $0.962$ for the kink, linear, and SBM designs, respectively. The
corresponding fractions of intervals equal to $[0,1]$ are $0.960$, $0.958$,
and $0.615$. By $n=2000$, the mean widths decrease to approximately $0.49$--$0.52$
and remain greater than the corresponding Gaussian widths.

\subsection{Supporting estimation diagnostics and bandwidth sensitivity}
Table~\ref{tab:rate-simulation} reports one-sided entrywise MSE at
$h_n=h_n^{\mathrm{bal}}$. The scaled column
$n^{2/3}\times\mathrm{MSE}$ permits comparison with the conditional
risk upper bound of order $n^{-2/3}$.

\input{generated_results/regenerated_20260831/tables/rate.tex}

For all three models, the displayed values of
$n^{2/3}\mathrm{MSE}$ lie approximately between $0.20$ and $0.29$.
This is consistent with the upper-bound scaling $O(n^{-2/3})$ at
$h_n\asymp n^{2/3}$, but does not establish a sharp constant or a
minimax rate. The reported LOO and fixed-cardinality ZLZ MSEs are similar across
the simulated settings.

For the kink model, the root mean squared conditional biases at
$n=200,500,1000,2000$ are $0.02577,0.01891,0.01283,$ and $0.00946$,
respectively. The corresponding values multiplied by $n^{1/3}$ remain bounded
across these four sample sizes, in agreement with the scale of the
$h_n/n$ term in the theoretical bias bound.

\paragraph{Neighborhood-size sensitivity.}

For the kink graphon at $n=1000$, an independent set of $100$ simulated
graphs varies the integer neighborhood size according to
\begin{equation}
\label{eq:bandwidth-sweep}
h_n
=
\operatorname{round}(c\,n^{2/3}),
\qquad
c\in\{0.5,1,2,4\}.
\end{equation}

\input{generated_results/regenerated_20260831/tables/bandwidth_sweep.tex}

The experiments underlying this table and Table~\ref{tab:rate-simulation}
use independent simulation runs. Their entries at $(n,h_n)=(1000,100)$
therefore need not agree exactly.

Over the displayed range, MSE decreases while mean standardized bias
increases from $0.166$ to $0.436$. Thus, in this experiment, larger
neighborhoods are associated with lower estimation error but greater
bias relative to the standard error. The theoretical balance in
Corollary~\ref{cor:bandwidth-risk} does not identify a finite-sample
MSE minimizer.

\section{Discussion and extensions}
\label{sec:discussion}

The inferential scope is determined by two separate requirements.
Leave-one-out selection makes the target-column observations
conditionally independent, while geometric localization controls the
difference between their conditional mean and $P_{ij}$. Accordingly,
empirical Bernstein inference covers the selected-neighborhood mean,
whereas Gaussian inference on $P_{ij}$ additionally requires negligible
bias on the standard-deviation scale. This explains the distinction
between the conditional-risk balance $h_n\asymp n^{2/3}$ and the
smooth-block inference condition $h_n=o(n^{2/3})$. On blocks with
identical profiles, within-block selection eliminates the bias and
permits $C_0\log n\le h_n=o(n)$ under the stated assumptions.

\subsection{Computational cost}
\label{subsec:computational-cost}
Forming $A^2$ by standard matrix multiplication costs
$\mathcal O(n^3)$. The reduced two-hop matrices need not be recomputed from
scratch: for every target column $j$,
\begin{equation}
\label{eq:rank-one-update}
\bigl(A^{(-j)}\bigr)^2
=
\bigl(A^2\bigr)_{-j,-j}
-
A_{-j,j}A_{j,-j}.
\end{equation}
Thus, each reduced square is obtained from $A^2$ by a rank-one correction.

The rank-one update avoids a separate matrix multiplication for each
target vertex, but full LOO inference remains more expensive than a
single ZLZ fit. After forming $A^2$, direct distance calculations cost
$\mathcal O(n^2)$ for one ordered target pair, $\mathcal O(n^3)$ for
one target row, and $\mathcal O(n^4)$ for all ordered pairs in a naive
implementation. Target columns can be processed in parallel. More
efficient implementations for full-matrix inference, such as landmark
or approximate nearest-neighbor methods, remain outside our scope.

\subsection{Limitations}

Inference on $P_{ij}$ requires quantitative within-block
identifiability and separation between distinct blocks in addition
to graphon regularity. Lipschitz continuity alone does not imply
these conditions, although some nonconstant Lipschitz graphons satisfy
them. For instance, $f(u,v)=a+buv$ with $a,b>0$ and $a+b<1$
satisfies regularity, nondegeneracy, and within-block identifiability
on a single smooth block. The localization theorem eliminates fixed
interior trimming but retains the structural assumptions. Moreover,
the finite-sample interval for $P_{ij}$ involves a generally unknown
bias constant, and the feasible Gaussian interval provides only
asymptotic pointwise coverage, not simultaneous coverage.

Our cross-validation identity justifies the expected criterion for a fixed
neighborhood size, but does not show that the data-selected minimizer has
an oracle risk bound or that confidence intervals remain valid after
selection. Establishing such results would require additional analysis
of the dependence among the cross-validation scores and the selected
estimator.

\subsection{Future work}
Extensions to directed, weighted, sparse, or degree-corrected networks
would require selection distances, variance bounds, and bias arguments
adapted to the corresponding edge models. Simultaneous inference for
multiple edge probabilities would require joint distributional
approximations and bias control at the relevant scale. Determining
whether the identifiability and block-separation conditions can be
verified from a single observed network remains an open problem.

\section*{Acknowledgments}

No funding was received for this work.
\section*{Declaration of interest}

The authors report no conflicts of interest. The authors alone are responsible for the content and writing of the paper.

\appendix

\section{A sufficient condition for two-hop identifiability}
\label{app:identifiability-notes}
\subsection{Blockwise form and a sufficient condition}
By symmetry, Assumption~\ref{as:hybrid} makes the integrand in
\eqref{eq:D2-definition} continuous in $w$ on each smooth block and constant
in $w$ on each block with identical probability profiles. Because the partition has finitely
many boundaries, the essential supremum in $D_2$ is equivalently
\begin{equation}
\label{eq:D2-blockwise}
D_2(u,u')
=
\max_{0\le m\le K-1}
\sup_{w\in I_m}
\left|
\int_0^1
\{f(u,x)-f(u',x)\}f(x,w)\,dx
\right|.
\end{equation}
Values of the integrand at the finitely many block boundaries do not affect
the essential supremum. Under the continuous latent-position design, latent
positions also fall on those boundaries with probability zero.

The finite-dimensionality and injectivity conditions below are sufficient
for Assumption~\ref{as:nondeg-block}(1).
For a smooth block $I_m$, define
\[
\mathcal H_m
:=
\operatorname{span}\bigl\{
f(u,\cdot)-f(u',\cdot):u,u'\in I_m
\bigr\}
\subset L^2[0,1],
\]
and let
\[
(T_f g)(w):=\int_0^1 g(x)f(x,w)\,dx.
\]

Suppose that $\mathcal H_m$ is finite dimensional and that
$g\mapsto (T_f g)|_{I_m}$ is injective on $\mathcal H_m$. Then there
exists $\kappa_m>0$ such that
\[
\sup_{w\in I_m}|(T_f g)(w)|
\ge
\kappa_m\|g\|_{L^2}
\qquad
\text{for every }g\in\mathcal H_m.
\]
The uniform lower bound follows from continuity and compactness. By symmetry
and Assumption~\ref{as:hybrid}, $T_fg$ is continuous on each block.
Moreover, because $0\le f\le1$, for all $g,h\in L^2[0,1]$,
\[
\sup_{w\in I_m}
\left|
T_f(g-h)(w)
\right|
\le
\|g-h\|_{L^2}.
\]
Hence
\(
g\longmapsto\sup_{w\in I_m}|(T_fg)(w)|
\)
is continuous in the $L^2$ norm. It is strictly positive on the
$L^2$-unit sphere of $\mathcal H_m$ by injectivity. Since this unit sphere
is compact, the displayed functional has a strictly positive minimum.
Consequently, Assumption~\ref{as:nondeg-block}(1) holds on $I_m$.

The within-block quantitative identifiability condition in
Assumption~\ref{as:nondeg-block}(1) converts a small two-hop discrepancy into
a small $L^2$ distance between graphon slices; the across-block separation
condition in Assumption~\ref{as:nondeg-block}(2) prevents vertices from
different blocks from being treated as neighbors. These two conditions are
needed for the bias analysis, not for the conditional-independence result.
Neither condition follows from Lipschitz regularity alone: in an
infinite-dimensional class, injectivity of the two-hop operator alone need
not give the positive uniform lower bound required by
Assumption~\ref{as:nondeg-block}(1).

\section{Classical ZLZ smoothing and ancillary point estimation}
\label{app:zlz-comparison}
\subsection{Classical ZLZ construction}

This appendix recalls the construction of \citet{zhang2017estimating},
abbreviated ZLZ. Unless another document is named, equation numbers in this
section refer to their main \emph{Biometrika} article. Their Supplementary
Material has a separate equation sequence.

Define the classical two-hop distance
$d_{\mathrm{ZLZ}}:V\times V\to\mathbb R_+$ by
\begin{equation}
\label{eq:dij}
d_{\mathrm{ZLZ}}(i,j):=
\begin{cases}
0, & i=j,\\
\max_{k\ne i,j}|M_{ik}-M_{jk}|, & i\ne j.
\end{cases}
\end{equation}
Here $M=A^2/n$ is the normalized two-hop matrix from
\eqref{eq:twohop-matrices}. The ZLZ main article introduces the latent
$L^2$ slice distance before equation~(5), expands its square in
equation~(5), bounds it in equation~(6), defines the empirical quantity
$\widetilde d^{\,2}(i,i')$ in equation~(7), and defines its neighborhood in
equation~(8). Since $A$ is symmetric,
\[
\frac{1}{n}\langle A_{i\cdot}-A_{i'\cdot},A_{k\cdot}\rangle
=M_{ik}-M_{i'k}.
\]
Thus $d_{\mathrm{ZLZ}}$ in \eqref{eq:dij} is the quantity denoted
$\widetilde d^{\,2}$ in equation~(7) of the main article, written using
$M$. The LOO distance in \eqref{eq:d-j} uses the same type of row
comparison after vertex $j$ has been removed.

ZLZ use $h$ for a quantile proportion. In Theorem~1 of their main article,
$h=C(n^{-1}\log n)^{1/2}$. Section~1 of their Supplementary Material reports
a simulation sensitivity study at $n=2000$ for four graphons, with
$C\in\{2^{-3},2^{-2},\ldots,2^3\}$. They use $C=1$ in their remaining
numerical results. To distinguish this proportion from our integer
neighborhood size $h_n$, let $\tau_n\in(0,1)$ denote the ZLZ proportion and
let $Q_i^{\mathrm{ZLZ}}(\tau_n)$ be the $\tau_n$-quantile of
\(
\{d_{\mathrm{ZLZ}}(i,k):k\in V\setminus\{i\}\}.
\)
The ZLZ neighborhood and one-sided smoother are
\begin{equation}
\label{eq:ZLZ-neighborhood}
\Ncal_i^{\mathrm{ZLZ}}
:=
\bigl\{
k\in V\setminus\{i\}:
d_{\mathrm{ZLZ}}(i,k)\le Q_i^{\mathrm{ZLZ}}(\tau_n)
\bigr\},
\end{equation}
and
\[
\widetilde P^{\mathrm{ZLZ}}_{ij}
:=
\frac{1}{|\Ncal_i^{\mathrm{ZLZ}}|}
\sum_{k\in\Ncal_i^{\mathrm{ZLZ}}}A_{kj}.
\]
For the simulations in Section~\ref{sec:sim}, we use a fixed-cardinality
version that selects the $h_n$ smallest values of
$d_{\mathrm{ZLZ}}(i,\cdot)$.

Because the ZLZ neighborhood uses the full adjacency matrix,
conditioning on it may induce dependence among the target-column edges
subsequently averaged. The following finite example demonstrates this
phenomenon for the fixed-cardinality version used in
Section~\ref{sec:sim}.

\subsection{A finite example of conditional dependence}
Let $n=4$ and suppose that
\[
P_{rs}=\frac12,
\qquad
1\le r<s\le4,
\]
so the six upper-triangular edges are independent
$\operatorname{Bernoulli}(1/2)$ variables. Take the smoothing vertex
$i=1$, the target column $j=4$, and let the fixed-cardinality ZLZ rule select
the two smallest distances, with ties broken by increasing vertex label as
in the simulations. Write
\[
a:=A_{12},\quad b:=A_{13},\quad c:=A_{14},\quad d:=A_{23},
\qquad
X:=A_{24},\quad Y:=A_{34}.
\]
Direct evaluation of $M=A^2/4$ gives
\begin{align*}
4d_{\mathrm{ZLZ}}(1,2)
&=
\max\bigl\{
|ad+cY-ab-XY|,
|aX+bY-ac-dY|
\bigr\},\\
4d_{\mathrm{ZLZ}}(1,3)
&=
\max\bigl\{
|bd+cX-ab-XY|,
|aX+bY-bc-dX|
\bigr\},\\
4d_{\mathrm{ZLZ}}(1,4)
&=
\max\bigl\{
|bd+cX-ac-dY|,
|ad+cY-bc-dX|
\bigr\}.
\end{align*}
For each fixed value of $(X,Y)$, enumeration of the $2^4$ configurations of
$(a,b,c,d)$ gives the following numbers of configurations for which the two
selected vertices are $\{2,3\}$:
\[
\begin{array}{c|cccc}
(X,Y) &(0,0)&(0,1)&(1,0)&(1,1)\\
\hline
\#\{(a,b,c,d):\Ncal_1^{\mathrm{ZLZ}}=\{2,3\}\}
&14&11&11&9.
\end{array}
\]
Let $E:=\{\Ncal_1^{\mathrm{ZLZ}}=\{2,3\}\}$. Since the $2^6$
edge configurations are equiprobable, $\PP(E)=45/64$, and the preceding
counts yield
\[
\PP(X=1\mid E)
=
\PP(Y=1\mid E)
=
\frac{20}{45}
=
\frac49,
\qquad
\PP(X=1,Y=1\mid E)
=
\frac9{45}
=
\frac15.
\]
Consequently,
\[
\PP(X=1,Y=1\mid E)
=
\frac15
\neq
\frac{16}{81}
=
\PP(X=1\mid E)\PP(Y=1\mid E).
\]
Thus, even conditional on the fixed matrix $P$ and the final selected
neighborhood, the selected target-column variables $A_{24}$ and $A_{34}$
are not independent. Centering them by $P_{24}$ and $P_{34}$ does not restore
independence. Hence the conditional empirical-Bernstein and Berry--Esseen
arguments used for the LOO estimator do not automatically apply to this
full-data ZLZ average. The example does not exclude distributional analysis of the
full-data estimator under different arguments.

\subsection{Ancillary symmetrized point estimate}
For completeness, define the symmetric point estimate
\[
\widehat P_{ij}
=
\frac{\tilde P_{ij}+\tilde P_{ji}}{2}.
\]
It respects the symmetry of the undirected graph but is not used in our
inferential construction. Using $P_{ji}=P_{ij}$ and
the convexity of $x\mapsto x^2$ gives
\[
(\widehat P_{ij}-P_{ij})^2
\le
\frac12(\tilde P_{ij}-P_{ij})^2
+
\frac12(\tilde P_{ji}-P_{ij})^2.
\]
Consequently, unconditional entrywise mean squared error bounds for
both one-sided directions transfer to the symmetrized estimator by taking
expectations. The same deterministic inequality also transfers realized
squared-error bounds. It does not, by itself, yield a conditional mean
squared error formula given both selected neighborhoods, nor does it
transfer the one-sided conditional concentration or normal-approximation
results. The estimators $\tilde P_{ij}$ and $\tilde P_{ji}$ are generally
dependent through their selected neighborhoods and shared edge variables.
Accordingly, all inferential results in the main text concern the
one-sided statistic $\tilde P_{ij}$.

\section{Proofs of the risk and inference results}
\label{app:main-proofs}
This appendix contains the conditional-independence, risk, and inference
proofs. Results involving localization and smoothing bias are proved
separately in Appendix~\ref{app:geometry-proofs}.
\begin{proof}[Proof of Proposition~\ref{prop:independence}]
Condition on $P$. Under the model \eqref{eq:model}, the upper-triangular
edge variables are mutually independent Bernoulli random variables with
the corresponding entries of $P$ as parameters. Define
\[
\mathcal F_{-j}
:=
\sigma\!\left(
A_{ab}:1\le a<b\le n,\ a\neq j,\ b\neq j
\right)
\]
and
\(
\mathcal F_j
:=
\sigma\!\left(A_{kj}:k\neq j\right).
\)
Under the model \eqref{eq:model}, the underlying edge collections that
generate $\mathcal F_{-j}$ and $\mathcal F_j$ are disjoint. Hence,
conditional on $P$,
\[
\mathcal F_{-j}
\perp\!\!\!\perp
\mathcal F_j.
\]
By construction, $A^{(-j)}$ is $\mathcal F_{-j}$-measurable. The
deterministic tie-breaking rule in \eqref{eq:neighborhood} makes
$\Ncal_i^{(-j)}$ a measurable function of $A^{(-j)}$; therefore
\(
\sigma(\Ncal_i^{(-j)})\subseteq\mathcal F_{-j}.
\)
Thus every event of the form $\{\Ncal_i^{(-j)}=S\}$ is conditionally
independent of $\mathcal F_j$. Consequently, for every realization $S$
with positive conditional probability given $P$, conditioning further on
$\{\Ncal_i^{(-j)}=S\}$ does not change the joint conditional law of the
target-column edges $(A_{kj})_{k\in S}$. Since these edges are mutually
independent Bernoulli variables conditional on $P$,
\[
\mathcal L\!\left(
(A_{kj})_{k\in S}
\,\middle|\,
\Ncal_i^{(-j)}=S,P
\right)
=
\bigotimes_{k\in S}\operatorname{Bernoulli}(P_{kj}).
\]
The centered variables therefore have conditional mean zero and are
supported in $[-1,1]$, completing the proof.
\end{proof}

\begin{proof}[Proof of Theorem~\ref{thm:loo-mse}]
Condition on $(\Ncal_i^{(-j)},P)$. Using the error decomposition \eqref{eq:error-decomposition},
we obtain
\[
\EE\!\left[(\tilde P_{ij}-P_{ij})^2 \,\middle|\, \Ncal_i^{(-j)},P\right]
=
\EE\!\left[(U_{ij}^{(-j)})^2 \,\middle|\, \Ncal_i^{(-j)},P\right]
+
2B_{ij}^{(-j)}\,\EE\!\left[U_{ij}^{(-j)} \,\middle|\, \Ncal_i^{(-j)},P\right]
+
(B_{ij}^{(-j)})^2.
\]
By Proposition~\ref{prop:independence},
\(
\EE\!\left[U_{ij}^{(-j)} \,\middle|\, \Ncal_i^{(-j)},P\right]=0,
\)
so the cross-term vanishes. Again by Proposition~\ref{prop:independence}, conditional on $(\Ncal_i^{(-j)},P)$ the summands are independent and mean zero. Therefore
\[
\EE\!\left[(U_{ij}^{(-j)})^2 \,\middle|\, \Ncal_i^{(-j)},P\right]
=
\Var\!\left(U_{ij}^{(-j)} \,\middle|\, \Ncal_i^{(-j)},P\right)
=
\frac{1}{h_n^2}\sum_{k\in\Ncal_i^{(-j)}}P_{kj}(1-P_{kj})
\le
\frac{1}{4h_n}.
\]
For the bias term,
\[
|B_{ij}^{(-j)}|
=
\left|
\frac{1}{h_n}\sum_{k\in\Ncal_i^{(-j)}} (P_{kj}-P_{ij})
\right|
\le
\frac{1}{h_n}\sum_{k\in\Ncal_i^{(-j)}} |P_{kj}-P_{ij}|
\le
\max_{k\in\Ncal_i^{(-j)}} |P_{kj}-P_{ij}|,
\]
hence
\(
(B_{ij}^{(-j)})^2
\le
\max_{k\in\Ncal_i^{(-j)}} (P_{kj}-P_{ij})^2.
\)
Combining the bounds gives
\[
\EE\!\left[(\tilde P_{ij}-P_{ij})^2 \,\middle|\, \Ncal_i^{(-j)},P\right]
\le
\frac{1}{4h_n}
+
\max_{k\in\Ncal_i^{(-j)}} (P_{kj}-P_{ij})^2,
\]
as claimed.
\end{proof}

\begin{proof}[Proof of Corollary~\ref{cor:bandwidth-risk}]
On $\Ecal_{\mathrm{master}}$, Theorem~\ref{thm:bias-goodset} gives
$|B_{ij}^{(-j)}|\le Cb_n(h_n)$. The exact conditional risk identity in
Theorem~\ref{thm:loo-mse}, together with
$P_{kj}(1-P_{kj})\le 1/4$, yields
\[
\EE\!\left[
(\tilde P_{ij}-P_{ij})^2
\,\middle|\,
\Ncal_i^{(-j)},P
\right]
=
\frac{1}{h_n^2}\sum_{k\in\Ncal_i^{(-j)}}P_{kj}(1-P_{kj})
+(B_{ij}^{(-j)})^2
\le \frac{1}{4h_n}+C^2b_n(h_n)^2.
\]
Absorbing $C^2$ into $C$ proves \eqref{eq:bandwidth-risk-bound}.
At $h_n\asymp n^{2/3}$,
$h_n/n\asymp n^{-1/3}$ and
$r_n=\sqrt{\log n/n}=o(n^{-1/3})$, so
$b_n(h_n)\asymp n^{-1/3}$. Hence the right-hand side is
$\mathcal O(n^{-2/3})$.
\end{proof}

\begin{proof}[Proof of Lemma~\ref{lem:empirical-bernstein}]
Theorem~11 of
\citet[Section~2, p.~4]{maurer2009empirical} applies to independent, not
necessarily identically distributed, random variables taking values in
$[0,1]$. Apply the theorem with failure probability $\alpha/2$ first to
$X_1,\ldots,X_m$ and then to $1-X_1,\ldots,1-X_m$. The two collections
have the same empirical variance $s^2$. A union bound therefore shows
that the two one-sided inequalities hold simultaneously with probability
at least $1-\alpha$, yielding
\eqref{eq:empirical-bernstein-two-sided}.
\end{proof}

\begin{proof}[Proof of Lemma~\ref{lem:var-lower}]
For each $k\in\Ncal_i^{(-j)}$, Assumption~\ref{as:nondeg} gives
$\delta\le P_{kj}\le1-\delta$, and therefore
\[
\delta(1-\delta)\le P_{kj}(1-P_{kj})\le \frac14.
\]
Summing these bounds over the $h_n$ selected vertices and dividing by
$h_n^2$ gives \eqref{eq:variance-bounds}. The bounds are uniform in
$(i,j)$ and in the selected neighborhood because the same $\delta$
applies to every edge.
\end{proof}

\begin{proof}[Proof of Theorem~\ref{thm:bias-goodset}]
Lemma~\ref{lem:containment} gives same-block containment for every selected neighbor on $\Ecal_{\mathrm{master}}$. On a block with identical probability profiles, Assumption~\ref{as:hybrid}(2) then gives $B_{ij}^{(-j)}=0$. On a smooth block, Lemma~\ref{lem:metric-to-slice} gives $|B_{ij}^{(-j)}|\le Cb_n(h_n)$. These statements hold uniformly over $(i,j)\in\Ecal_n$ on the combined event, proving the theorem.
\end{proof}
\begin{proof}[Proof of Theorem~\ref{thm:finite-ci}]
Let $\mathcal F_{ij}=\sigma(\boldsymbol\xi,A^{(-j)})$, where
$\boldsymbol\xi=(\xi_1,\ldots,\xi_n)$. The selected neighborhood
is $\mathcal F_{ij}$-measurable and, conditionally on
$\mathcal F_{ij}$, the selected target-column edges are independent
Bernoulli variables with respective probabilities $P_{kj}$.
Lemma~\ref{lem:empirical-bernstein} therefore gives
\[
\PP\bigl(|U_{ij}^{(-j)}|\le w_{ij}^{(B)}(\alpha)
\mid\mathcal F_{ij}\bigr)\ge1-\alpha.
\]
Because $\sigma(\boldsymbol\xi,\Ncal_i^{(-j)})\subseteq
\mathcal F_{ij}$, taking conditional expectations with respect to
$\sigma(\boldsymbol\xi,\Ncal_i^{(-j)})$ proves the first
conditional-coverage assertion.
Let $\mathcal B_{ij}=\{|B_{ij}^{(-j)}|\le Cb_n(h_n)\}$.
This event is $\mathcal F_{ij}$-measurable, and
Theorem~\ref{thm:bias-goodset} implies
$\PP(\mathcal B_{ij}^c)\le Cn^{-c}$.
On the intersection of the two events,
$|\tilde P_{ij}-P_{ij}|\le w_{ij}^{(B)}(\alpha)+Cb_n(h_n)$.
The union bound and the fact that both the estimator and target
belong to $[0,1]$ yield \eqref{eq:bernstein-coverage}.
\end{proof}

\begin{proof}[Proof of Theorem~\ref{thm:be} and Theorem~\ref{thm:full-clt}]
For the proof of Theorem~\ref{thm:be}, fix arbitrary distinct indices
$i\neq j$. Conditional on $(\Ncal_i^{(-j)},P)$, let
\[
X_k:=A_{kj}-P_{kj},
\qquad
k\in\Ncal_i^{(-j)}.
\]

These variables are independent, mean-zero, and bounded. With
$S_{ij,n}$ and $\sigma_{ij,n}$ defined in
\eqref{eq:standardization-identity}, Assumption~\ref{as:nondeg} gives
$\sigma_{ij,n}^2\ge h_n\delta(1-\delta)$. Since $|X_k|\le1$,
\[
\rho_3
:=
\sum_{k\in\Ncal_i^{(-j)}}
\EE\!\left[|X_k|^3\,\middle|\,\Ncal_i^{(-j)},P\right]
\le h_n.
\]
Hence
\[
\frac{\rho_3}{\sigma_{ij,n}^3}
\le
\frac{h_n}{\{h_n\delta(1-\delta)\}^{3/2}}
=
O(h_n^{-1/2}).
\]

The classical Berry--Esseen theorem for independent, not necessarily
identically distributed, triangular arrays \citep{Petrov1975}, together
with the standardization identity \eqref{eq:standardization-identity},
therefore gives the conditional bound \eqref{eq:stochastic-be-bound}.
By the tower property, for each $x\in\R$,
\[
\PP\!\left(\frac{U_{ij}^{(-j)}}{\sqrt{V_{ij}}}\le x\right)
=
\EE\!\left[
\PP\!\left(\frac{U_{ij}^{(-j)}}{\sqrt{V_{ij}}}\le x \,\middle|\, \Ncal_i^{(-j)},P\right)
\right].
\]
Therefore, for each $x\in\R$,
\[
\begin{aligned}
\left|
\PP\!\left(\frac{U_{ij}^{(-j)}}{\sqrt{V_{ij}}}\le x\right)-\Phi(x)
\right|
&=
\left|
\EE\!\left[
\PP\!\left(\frac{U_{ij}^{(-j)}}{\sqrt{V_{ij}}}\le x \,\middle|\, \Ncal_i^{(-j)},P\right)
-\Phi(x)
\right]
\right| \\
&\le
\EE\!\left[
\left|
\PP\!\left(\frac{U_{ij}^{(-j)}}{\sqrt{V_{ij}}}\le x \,\middle|\, \Ncal_i^{(-j)},P\right)
-\Phi(x)
\right|
\right]
\le C h_n^{-1/2}.
\end{aligned}
\]
Taking the supremum over $x\in\R$ yields the unconditional
$Ch_n^{-1/2}$ bound stated in Theorem~\ref{thm:be}.

For Theorem~\ref{thm:full-clt}, use the explicit bound \eqref{eq:full-be-bound}. Under the neighborhood-size condition \eqref{eq:geometric-regime}, $h_n\to\infty$, and the neighborhood-size condition \eqref{eq:clt-neighborhood-size-condition} makes the bias term vanish.
Hence every term on the right-hand side converges to zero, proving the
unconditional central limit theorem.
\end{proof}

\begin{proof}[Proof of Corollary~\ref{cor:full-be}]
Let $\mathcal F_{ij}=\sigma(\boldsymbol\xi,A^{(-j)})$ and set
$Z_{ij}=U_{ij}^{(-j)}/\sqrt{V_{ij}}$ and
$\mu_{ij}=B_{ij}^{(-j)}/\sqrt{V_{ij}}$.
The neighborhood, the bias, and the conditional variance are
$\mathcal F_{ij}$-measurable. Given this sigma-field, the selected target-column observations are
independent Bernoulli variables with respective means $P_{kj}$.
Repeating the Berry--Esseen calculation in the proof of
Theorem~\ref{thm:be} gives the conditional Kolmogorov bound
$Ch_n^{-1/2}$ for $Z_{ij}$, with a constant depending only on
$\delta$.
Since the standard normal distribution function is Lipschitz,
\[
\sup_x\bigl|\PP(Z_{ij}+\mu_{ij}\le x\mid\mathcal F_{ij})
-\Phi(x)\bigr|
\le Ch_n^{-1/2}+\min\{1,(2\pi)^{-1/2}|\mu_{ij}|\}.
\]
The bias-control event
$\mathcal B_{ij}=\{|B_{ij}^{(-j)}|\le Cb_n(h_n)\}$
is $\mathcal F_{ij}$-measurable and satisfies
$\PP(\mathcal B_{ij}^c)\le\PP(\Ecal_{\mathrm{master}}^c)
\le Cn^{-c}$. On $\mathcal B_{ij}$, the lower variance bound
\eqref{eq:variance-bounds} gives
$|\mu_{ij}|\le C\sqrt{h_n}\,b_n(h_n)$.
Average the conditional bound over $\mathcal F_{ij}$ and bound
the probability of $\mathcal B_{ij}^c$ separately to obtain
\eqref{eq:full-be-bound}.
\end{proof}

\subsection{Consistency of the empirical variance estimator}
Feasible normalization requires that the empirical variance estimate
converge to the conditional variance in relative error.

\begin{lemma}[Consistency of the empirical variance normalization]
\label{lem:emp-var-consistency}
Suppose Assumptions~\ref{as:hybrid}, \ref{as:nondeg}, and
\ref{as:nondeg-block} hold, and let $h_n$ satisfy
\eqref{eq:geometric-regime}. Uniformly over deterministic pairs
$i\ne j$, the variance estimator obeys
\begin{equation}
\label{eq:variance-ratio-consistency}
\frac{\widehat V_{ij}^{\mathrm{emp}}}{V_{ij}}
=1+\mathcal O_{\PP}\bigl(h_n^{-1/2}+b_n(h_n)^2+h_n^{-1}\bigr)
=1+o_{\PP}(1).
\end{equation}
Furthermore,
\[
\PP\bigl(\widehat V_{ij}^{\mathrm{emp}}=0
\mid\Ncal_i^{(-j)},P\bigr)\le2(1-\delta)^{h_n},
\]
so the feasible variance is positive with probability tending
to one. Here uniformity means that the probability bounds defining
$\mathcal O_{\PP}$ are independent of the deterministic ordered
pair; it does not assert the same rate for the maximum over all pairs.
\end{lemma}

\begin{proof}[Proof of Lemma~\ref{lem:emp-var-consistency}]
Conditional on $(\Ncal_i^{(-j)},P)$, write
\[
p_k:=P_{kj},
\qquad
\bar p:=\frac1{h_n}\sum_{k\in\Ncal_i^{(-j)}}p_k,
\qquad
\bar v:=\frac1{h_n}\sum_{k\in\Ncal_i^{(-j)}}p_k(1-p_k).
\]
Then $V_{ij}=\bar v/h_n$. Let
\[
\bar A:=\tilde P_{ij}
=\frac1{h_n}\sum_{k\in\Ncal_i^{(-j)}}A_{kj},
\qquad
g(x):=x(1-x).
\]
Conditional independence and Hoeffding's inequality give
$\bar A-\bar p=\mathcal O_{\PP}(h_n^{-1/2})$. Since $g$ is Lipschitz on
$[0,1]$,
\(
g(\bar A)-g(\bar p)=\mathcal O_{\PP}(h_n^{-1/2}).
\)
For binary observations,
\[
s_{ij}^2=\frac{h_n}{h_n-1}g(\bar A),
\qquad
g(\bar p)-\bar v
=\frac1{h_n}\sum_{k\in\Ncal_i^{(-j)}}(p_k-\bar p)^2.
\]
On $\Ecal_{\mathrm{master}}$, Corollary~\ref{cor:all-slice-close} gives
\(
\max_{k\in\Ncal_i^{(-j)}}|p_k-P_{ij}|\le Cb_n(h_n),
\)
and hence
\(
0\le g(\bar p)-\bar v\le4C^2b_n(h_n)^2.
\)
Conditional Hoeffding's inequality gives, uniformly over deterministic
ordered pairs and every $M>0$,
\[
\PP\!\left(\left|\bar A-\bar p\right|>
M h_n^{-1/2}\,\middle|\,\Ncal_i^{(-j)},P\right)
\le2e^{-2M^2}.
\]
Moreover, the uniform profile bound implies that the displayed
heterogeneity term is at most $4C^2b_n(h_n)^2$ except on an event
of probability $Cn^{-c}$. Combining these inequalities with
$s_{ij}^2=h_n g(\bar A)/(h_n-1)$ gives, for an appropriate
constant $C_1$,
\[
\begin{aligned}
&\sup_{i\ne j}\PP\!\left(
|s_{ij}^2-\bar v|>
C_1\{M h_n^{-1/2}+h_n^{-1}+b_n(h_n)^2\}\right)\\
&\qquad\le2e^{-2M^2}+Cn^{-c}.
\end{aligned}
\]
Since $\bar v\ge\delta(1-\delta)$ under
Assumption~\ref{as:nondeg}, division by $\bar v$ proves
\eqref{eq:variance-ratio-consistency}. This argument shows that
uniformity refers to the constants in pointwise probability bounds;
it does not assert a corresponding bound for the maximum of the
variance ratios.

Finally, $\widehat V_{ij}^{\mathrm{emp}}=0$ exactly when all selected binary
observations are zero or all are one. Conditional on
$(\Ncal_i^{(-j)},P)$,
\[
\begin{aligned}
\PP(s_{ij}^2=0\mid\Ncal_i^{(-j)},P)
&=
\prod_{k\in\Ncal_i^{(-j)}}(1-P_{kj})
+
\prod_{k\in\Ncal_i^{(-j)}}P_{kj}\\
&\le2(1-\delta)^{h_n}.
\end{aligned}
\]
This probability tends to zero because $h_n\to\infty$.
\end{proof}

\begin{proof}[Proof of Corollary~\ref{cor:feasible-normalization}]
Theorem~\ref{thm:full-clt} gives the oracle-normalized central limit
theorem, while Lemma~\ref{lem:emp-var-consistency} gives
$\widehat V_{ij}^{\mathrm{emp}}/V_{ij}\xrightarrow{\PP}1$.
On the event $\widehat V_{ij}^{\mathrm{emp}}>0$, Slutsky's theorem
therefore gives the claimed studentized limit. The complementary
event has probability tending to zero, so the convention on that
event does not affect weak convergence. Since $P_{ij}\in[0,1]$,
truncating the interval to $[0,1]$ leaves its coverage event
unchanged. Continuity of the standard normal law proves
\eqref{eq:gaussian-coverage}.
\end{proof}

\subsection{Proofs of the additional Gaussian inference results}
\begin{proof}[Proof of Theorem~\ref{thm:identical-block-clt}]
Let $\mathcal F_{ij}=\sigma(\boldsymbol\xi,A^{(-j)})$ and define
$\mathcal B_{ij}^{=}$ to be the $\mathcal F_{ij}$-measurable
event that $P_{kj}=P_{ij}$ for every
$k\in\Ncal_i^{(-j)}$. On the global localization event
$\Ecal_{\mathrm{master}}$ and on $H_i$, every selected neighbor lies
in the same identical-profile block as $\xi_i$; hence
$\mathcal B_{ij}^{=}$ occurs and $B_{ij}^{(-j)}=0$.
Therefore
\[
\PP\bigl((\mathcal B_{ij}^{=})^c\mid H_i\bigr)
\le\frac{\PP(\Ecal_{\mathrm{master}}^c)}{\pi_{\mathrm{const}}}
\le C\pi_{\mathrm{const}}^{-1}n^{-c}.
\]
Conditional on $\mathcal F_{ij}$, the selected target-column edges
are independent Bernoulli variables with parameters $P_{kj}$; their
centered versions are independent and mean zero. Applying the same
Berry--Esseen calculation as in the proof of
Theorem~\ref{thm:be} under this sigma-field yields a conditional
Kolmogorov bound $Ch_n^{-1/2}$, regardless of whether the global
localization event holds. Since $H_i$ belongs to $\mathcal F_{ij}$,
we may average this bound under $\PP(\,\cdot\mid H_i)$. On $\mathcal B_{ij}^{=}$, the full
estimation error equals the centered stochastic error. On its
complement, any distribution-function discrepancy is at most one.
Integrating the conditional approximation, first under $H_i$ and
then over $\mathcal F_{ij}$, proves \eqref{eq:constant-block-be}.

To establish feasible variance consistency, write
$p_k=P_{kj}$, $\bar p=h_n^{-1}\sum p_k$, and
$\bar v=h_n^{-1}\sum p_k(1-p_k)$, where the sums range over
$\Ncal_i^{(-j)}$. As in Lemma~\ref{lem:emp-var-consistency},
\[
s_{ij}^2-\bar v
=\frac{h_n}{h_n-1}\bar A(1-\bar A)-\bar v,
\qquad \bar A=\tilde P_{ij}.
\]
On $\mathcal B_{ij}^{=}$, all $p_k=P_{ij}$, so
$\bar p=P_{ij}$, $\bar v=P_{ij}(1-P_{ij})$, and
$\bar p(1-\bar p)=\bar v$.
Conditional Hoeffding concentration and the Lipschitz property of
$x(1-x)$ on $[0,1]$ therefore give
\[
\frac{\widehat V_{ij}^{\mathrm{emp}}}{V_{ij}}-1
=\mathcal O_{\PP}(h_n^{-1/2}+h_n^{-1})
\quad\text{under }\PP(\,\cdot\mid H_i).
\]
The exceptional-event probability has already been bounded above,
and $V_{ij}\ge\delta(1-\delta)/h_n$. The probability of a zero
empirical variance is at most $2(1-\delta)^{h_n}$ conditional on
the neighborhood and latent probabilities. As $h_n\to\infty$,
Slutsky's theorem under the law given $H_i$ proves the feasible
limit. The interval assertion follows from continuity of the
standard normal distribution.
\end{proof}

\begin{proof}[Proof of Theorem~\ref{thm:latent-selected-pair}]
Write $\boldsymbol\xi=(\xi_1,\ldots,\xi_n)$ and define
\[
\mathcal F_n:=\sigma\!\left(\boldsymbol\xi,\,
\mathbf1_{\{u\ne J_n,\,v\ne J_n\}}A_{uv}:
1\le u<v\le n\right).
\]
As $(I_n,J_n)$ is $\sigma(\boldsymbol\xi)$-measurable,
conditioning on $\mathcal F_n$ fixes the indices, the target-deleted
adjacency matrix, and the selected neighborhood. Conditional on
$\mathcal F_n$, the selected edges incident to $J_n$ are
independent Bernoulli variables with success probabilities
$P_{k,J_n}$. Their standardized centered sum satisfies
\eqref{eq:stochastic-be-bound} with a constant independent of the
realized latent positions and selected pair.

The global event $\Ecal_{\mathrm{master}}$ yields
$|B_{ij}^{(-j)}|\le Cb_n(h_n)$ simultaneously for all ordered
pairs, and in particular for $(I_n,J_n)$. Define the
$\mathcal F_n$-measurable event
\[
\mathcal B_n=
\{|B_{I_nJ_n}^{(-J_n)}|\le Cb_n(h_n)\}.
\]
Then $\PP(\mathcal B_n^c)\le Cn^{-c}$. On $\mathcal B_n$,
the bias divided by the oracle standard deviation is at most
$C\sqrt{h_n}\,b_n(h_n)$ by Lemma~\ref{lem:var-lower}.
The same conditional distribution-function comparison as in the
proof of Corollary~\ref{cor:full-be}, followed by averaging over
$\mathcal F_n$, proves \eqref{eq:latent-target-be}.
For the random target pair, put
\[
\bar A_n:=\tilde P_{I_nJ_n},\quad
\bar p_n:=\frac{1}{h_n}
\sum_{k\in\Ncal_{I_n}^{(-J_n)}}P_{kJ_n},\quad
\bar v_n:=\frac{1}{h_n}
\sum_{k\in\Ncal_{I_n}^{(-J_n)}}P_{kJ_n}(1-P_{kJ_n}).
\]
Given $\mathcal F_n$, the selected observations are independent
and bounded. Hoeffding's inequality therefore gives, for every
$M>0$,
\[
\PP\!\left(\left|\bar A_n-\bar p_n\right|>
\frac{M}{\sqrt{h_n}}\,\middle|\,\mathcal F_n\right)
\le 2e^{-2M^2}.
\]
On $\Ecal_{\mathrm{master}}$, the uniform profile bound implies
\[
0\le \bar p_n(1-\bar p_n)-\bar v_n
=\frac{1}{h_n}\sum_{k\in\Ncal_{I_n}^{(-J_n)}}
(P_{kJ_n}-\bar p_n)^2\le Cb_n(h_n)^2.
\]
Using the binary sample-variance identity and
$\bar v_n\ge\delta(1-\delta)$ yields
\[
\frac{\widehat V_{I_nJ_n}^{\mathrm{emp}}}{V_{I_nJ_n}}-1
=\mathcal O_{\PP}\!\left(
 h_n^{-1/2}+h_n^{-1}+b_n(h_n)^2\right).
\]
Here the exceptional event has probability at most $Cn^{-c}$,
and the conditional concentration bound is uniform over the
realized latent positions and target pair. This directly justifies
substitution of the latent-selected pair rather than relying on
pointwise uniformity alone. Slutsky's theorem and the zero-variance
convention give feasible normality and interval coverage.
\end{proof}

\subsection*{Conditional independence in the cross-validation score}
For a deterministic candidate $h$, the expectation identity
\eqref{eq:cv-risk-identity} follows from independence of the held-out
edge and its corresponding predictor.
Conditional on $P$, $\Ncal_i^{(-j)}(h)$ is measurable with respect to the
edges having neither endpoint equal to $j$ and is therefore independent
of the target-column edge vector
\[
\{A_{rj}:r\in[n]\setminus\{j\}\}.
\]
The predictor $\tilde P_{ij}^{(-j)}(h)$ is not measurable with respect to
$A^{(-j)}$ alone because, after neighborhood selection, it uses the
selected target-column observations
\[
\{A_{kj}:k\in\Ncal_i^{(-j)}(h)\}.
\]
Since $i\notin\Ncal_i^{(-j)}(h)$, however, the predictor does not use the
held-out edge $A_{ij}$. More precisely,
$\tilde P_{ij}^{(-j)}(h)$ is measurable with respect to
\[
\sigma\!\left(
A^{(-j)},\{A_{kj}:k\neq i,j\}
\right).
\]
Conditional edge independence therefore gives
\[
A_{ij}
\perp\!\!\!\perp
\left(
\Ncal_i^{(-j)}(h),
\tilde P_{ij}^{(-j)}(h)
\right)
\mid P.
\]
In particular,
\[
A_{ij}
\perp\!\!\!\perp
\tilde P_{ij}^{(-j)}(h)
\mid
\bigl(\Ncal_i^{(-j)}(h),P\bigr).
\]

\section{Auxiliary results for neighborhood selection and bias control}
\label{app:geometry-lemmas}
The localization argument combines two-hop concentration, an adequate
number of within-block candidate vertices, and separation of distinct
blocks. We place these requirements on one high-probability event so that
the resulting bias bounds hold simultaneously over ordered target pairs.

For the geometric analysis, let $\Ecal_{\mathrm{conc}}$ denote the event
on which the uniform bound of Lemma~\ref{lem:twohop-concentration} holds,
let $\Ecal_{\mathrm{int}}$ denote the event of
Lemma~\ref{lem:integral-to-sum}, and let $\Ecal_{\mathrm{fill}}$ denote the
event of Lemma~\ref{lem:max-to-sup}. In addition, with the candidate sets
$S_{ij}$ defined in \eqref{eq:candidate-set} below, let
\[
\Ecal_{\mathrm{cand}}
:=\left\{|S_{ij}|\ge h_n\text{ for all }(i,j)\in\Ecal_n\right\}.
\]
Define the preliminary geometric approximation event
\[
\Ecal_0
:=
\Ecal_{\mathrm{conc}}
\cap
\Ecal_{\mathrm{int}}
\cap
\Ecal_{\mathrm{fill}},
\]
and then define the combined event
\begin{equation}
\label{eq:master-event}
\Ecal_{\mathrm{master}}
:=
\Ecal_0\cap\Ecal_{\mathrm{cand}}.
\end{equation}
The across-block separation used later is a deterministic condition in
Assumption~\ref{as:nondeg-block}(2), so no additional random separation
event is required.

All probability statements are taken jointly over
$(\xi_1,\dots,\xi_n)$ and $A$. When analyzing empirical fluctuations, we
condition on the latent variables to isolate the randomness in $A$.

\begin{remark}[Combined high-probability event]
\label{rem:master-event}
By Lemmas~\ref{lem:twohop-concentration},
\ref{lem:integral-to-sum}, \ref{lem:max-to-sup}, and
\ref{lem:candidates}, together with a union bound, there exist constants
$C,c>0$ such that
\[
\PP(\Ecal_{\mathrm{master}})
\ge
1-Cn^{-c}.
\]
\end{remark}
Throughout Appendix~\ref{app:geometry-lemmas}, all deterministic
geometric inequalities are understood to hold on
$\Ecal_{\mathrm{master}}$. The event controls every ordered target
pair in $\Ecal_n$ simultaneously; no interior eligibility event is used.

\subsection{Empirical distance of nearby candidates} \label{subsec:empirical-concentration}
\begin{lemma}[Population analogue of the reduced two-hop matrix]
\label{lem:population-twohop}
Fix distinct indices $i,\ell,j$.

Conditional on the latent variables, there exist constants $C,c>0$ such
that, with conditional probability at least $1-Cn^{-c}$ over the edge
variables,
\[
\left|
M^{(-j)}_{i\ell}
-
\frac{1}{n-1}
\sum_{m\notin\{i,\ell,j\}}
P_{im}P_{m\ell}
\right|
\le
C\sqrt{\frac{\log n}{n}}
\]
uniformly over all admissible triples $(i,\ell,j)$.

Separately, with probability at least $1-Cn^{-c}$ over the latent
variables,
\[
\left|
\frac{1}{n-1}
\sum_{m\notin\{i,\ell,j\}}
P_{im}P_{m\ell}
-
\int_0^1
f(\xi_i,x)f(x,\xi_\ell)\,dx
\right|
\le
C\sqrt{\frac{\log n}{n}}
\]
uniformly over all admissible triples.

Consequently, under the joint law of the latent variables and edges,
\[
M^{(-j)}_{i\ell}
=
\int_0^1
f(\xi_i,x)f(x,\xi_\ell)\,dx
+
\mathcal O_{\PP}\!\left(
\sqrt{\frac{\log n}{n}}
\right)
\]
uniformly over all admissible indices.
\end{lemma}

\begin{lemma}[Two-hop concentration]
\label{lem:twohop-concentration}
There exist constants $C,c>0$ such that, with probability at least
$1-Cn^{-c}$,
\[
\max_{\substack{i,k,j\ \mathrm{distinct}\\
\ell\notin\{i,k,j\}}}
\left|
\left(M^{(-j)}_{i\ell}-M^{(-j)}_{k\ell}\right)
-
\frac{1}{n-1}
\sum_{m\notin\{i,k,\ell,j\}}
(P_{im}-P_{km})P_{m\ell}
\right|
\le
C\left(\frac{\log n}{n}\right)^{1/2}.
\]
\end{lemma}

The population counterpart of the finite two-hop difference is
\begin{equation}
\label{eq:twohop-difference-function}
g_{ik}(w):=\int_0^1\{f(\xi_i,x)-f(\xi_k,x)\}f(x,w)\,dx.
\end{equation}
\begin{lemma}[Population to empirical two-hop approximation]
\label{lem:integral-to-sum}
Under Assumption~\ref{as:hybrid}, there exist constants $C,c>0$ such that,
with probability at least $1-Cn^{-c}$ over the latent variables,
uniformly over all admissible distinct indices $(i,k',\ell,j)$,
\[
\left|
\frac{1}{n-1}
\sum_{m\notin\{i,k',\ell,j\}}
(P_{im}-P_{k'm})P_{m\ell}
-
\int_0^1
\{f(\xi_i,x)-f(\xi_{k'},x)\}f(x,\xi_\ell)\,dx
\right|
\le
C\left(\frac{\log n}{n}\right)^{1/2}.
\]
\end{lemma}
The empirical LOO distance uses a maximum over sampled latent positions,
whereas the population discrepancy involves blockwise suprema. The
following result bounds the error introduced by this discretization.

\begin{lemma}[Discrete maximum and supremum over each block]
\label{lem:max-to-sup}
Under Assumption~\ref{as:hybrid}, for distinct indices $i,k',j$, let
$g_{ik'}$ be the function defined in
\eqref{eq:twohop-difference-function}. Then there exist constants
$C,c>0$ such that, with probability at least $1-Cn^{-c}$, all inner
maxima in the following display are over nonempty index sets and,
uniformly over all admissible triples $(i,k',j)$,
\[
\max_{0\le m\le K-1}
\left|
\max_{\substack{\ell:\,\xi_\ell\in I_m\\
\ell\notin\{i,k',j\}}}
|g_{ik'}(\xi_\ell)|
-
\sup_{w\in I_m}|g_{ik'}(w)|
\right|
\le
C\frac{\log n}{n}.
\]
\end{lemma}

\begin{corollary}[Uniform empirical-to-population approximation]
\label{cor:empirical-to-integral}
Under Assumption~\ref{as:hybrid} and the model \eqref{eq:model}, with probability at least $1-Cn^{-c}$, we have uniformly for all admissible indices $j$ and all distinct $i, k', \ell \in V \setminus \{j\}$:
\[
\begin{aligned}
&\left| \left( M^{(-j)}_{i\ell} - M^{(-j)}_{k'\ell} \right) - \int_0^1 (f(\xi_i,x)-f(\xi_{k'},x))f(x,\xi_\ell)\,dx \right| \\
&\quad \le C \left(\frac{\log n}{n}\right)^{1/2}.
\end{aligned}
\]
\end{corollary}

\begin{lemma}[Uniform approximation of the LOO distance by $D_2$]
\label{lem:distance-approximation}
On the event $\Ecal_0$, there exists a constant $C>0$ such that,
uniformly over all distinct indices $i,k,j$,
\begin{equation}
\label{eq:distance-approximation}
\left|
 d^{(-j)}(i,k)-D_2(\xi_i,\xi_k)
\right|
\le
C r_n.
\end{equation}
\end{lemma}

To control the selected-neighborhood radius, we construct a theoretical
set of candidates in the same block as the smoothing vertex. Let
$(i,j)\in\Ecal_n$, let $I_m$ contain $\xi_i$, and take
\begin{equation}
\label{eq:candidate-radius}
t_n=a\frac{h_n}{n},\qquad a>1.
\end{equation}
Define
\begin{equation}
\label{eq:candidate-set}
S_{ij}:=
\begin{cases}
\{k\notin\{i,j\}:\xi_k\in I_m,
\ |\xi_k-\xi_i|\le t_n\},&I_m\text{ smooth},\\[2pt]
\{k\notin\{i,j\}:\xi_k\in I_m\},
&I_m\text{ has identical profiles}.
\end{cases}
\end{equation}
The smooth-block candidate interval is intersected with the block,
so it remains valid when $\xi_i$ lies arbitrarily close to an endpoint.
These candidate sets are used only in the analysis and do not enter the
observable estimator.

\begin{lemma}[Uniform candidate count and two-hop distance]
\label{lem:candidates}
Suppose Assumption~\ref{as:hybrid} and
\eqref{eq:geometric-regime} hold, and fix $a>1$ in
\eqref{eq:candidate-radius}. If $C_0$ is sufficiently large,
there are constants $C,c>0$ such that
\[
\PP\!\left(\min_{(i,j)\in\Ecal_n}|S_{ij}|<h_n\right)
\le Cn^{-c}.
\]
On the uniform empirical-to-population approximation event
$\Ecal_0$, every $k\in S_{ij}$ satisfies
\[
d^{(-j)}(i,k)\le C\left(\frac{h_n}{n}+r_n\right)
=Cb_n(h_n).
\]
The estimates hold simultaneously over all ordered pairs,
without interior trimming.
\end{lemma}

\subsection{Functional (slice) closeness of selected neighbors} \label{subsec:neighborhood-containment}
Control of the population two-hop discrepancy for selected neighbors,
combined with across-block separation and within-block identifiability,
yields neighborhood containment and closeness of connection profiles.

\begin{lemma}[Two-hop discrepancy of selected neighbors]
\label{lem:slice-close}
Fix $(i,j)\in \Ecal_n$ and let $\Ncal_i^{(-j)}$ be the LOO
neighborhood of size $h_n$.  
On the combined event $\Ecal_{\text{master}}$ defined in
\eqref{eq:master-event}, the following holds:
for every $k\in\Ncal_i^{(-j)}$, the discrepancy in \eqref{eq:D2-definition} satisfies
\begin{equation}
\label{eq:selected-D2-bound}
D_2(\xi_i,\xi_k)\le C b_n(h_n).
\end{equation}
\end{lemma}

Across-block separation converts Lemma~\ref{lem:slice-close} into block
containment.

\begin{lemma}[Neighborhood containment]
\label{lem:containment}
Fix $(i,j)\in\Ecal_n$ and suppose $\xi_i\in I_m$. Under
Assumption~\ref{as:nondeg-block}(2), on $\Ecal_{\mathrm{master}}$ and for
all sufficiently large $n$, every cross-block vertex has empirical LOO
distance bounded below by a positive constant, whereas the $h_n$
within-block candidate vertices of Lemma~\ref{lem:candidates} have distance
$O(b_n(h_n))=o(1)$. Consequently,
\[
\xi_k\in I_m
\qquad\text{for every }k\in\Ncal_i^{(-j)}.
\]
\end{lemma}
\begin{lemma}[Two-hop discrepancy implies $L^2$ slice closeness]
\label{lem:twohop-to-slice}
Let $i\in[n]$ lie in a smooth block $I_m$, and let
$k\in\Ncal_i^{(-j)}$. Under Assumption~\ref{as:nondeg-block}, on
$\Ecal_{\mathrm{master}}$,
\[
\|f(\xi_i,\cdot)-f(\xi_k,\cdot)\|_{L^2}
\le
C b_n(h_n).
\]
\end{lemma}

\subsection{From slice closeness to latent closeness on smooth blocks} \label{subsec:bias-inversion}
The lower within-block Lipschitz bound converts closeness of graphon
slices into closeness of latent positions. Pointwise Lipschitz regularity
then controls the edge-probability differences entering the smoothing
bias.
\begin{lemma}[Quantitative inversion on smooth blocks]
\label{lem:injectivity}
Under Assumption~\ref{as:nondeg-block}(1), for every smooth block $I_m$
and every $u,u'\in I_m$,
\[
\|f(u,\cdot)-f(u',\cdot)\|_{L^2}
\le
\kappa_m^{-1}
\sup_{w\in I_m}
\left|
\int_0^1
\{f(u,x)-f(u',x)\}f(x,w)\,dx
\right|.
\]
\end{lemma}

Combining Lemma~\ref{lem:injectivity} with the within-block
Lipschitz bounds in Assumption~\ref{as:hybrid}(1) yields latent
localization and entrywise bias control.
\begin{lemma}[Latent localization and bias control]
\label{lem:metric-to-slice}
Suppose Assumptions~\ref{as:hybrid} and~\ref{as:nondeg-block} hold. On
$\Ecal_{\mathrm{master}}$, there exists $C>0$ such that the following
holds for every $(i,j)\in\Ecal_n$. Let $I_m$ be the block
containing $\xi_i$.

If $I_m$ is a smooth block, then every
$k\in\Ncal_i^{(-j)}$ belongs to $I_m$ and
\[
|\xi_k-\xi_i|
\le
C b_n(h_n).
\]
Consequently,
\[
|B_{ij}^{(-j)}|
\le
C b_n(h_n).
\]

If $I_m$ is a block with identical probability profiles, then every
$k\in\Ncal_i^{(-j)}$ belongs to $I_m$ and
\[
B_{ij}^{(-j)}=0.
\]
\end{lemma}

\begin{corollary}[Pointwise closeness of selected neighbors]
\label{cor:all-slice-close}
On $\Ecal_{\mathrm{master}}$, let $(i,j)\in\Ecal_n$ and
$k\in\Ncal_i^{(-j)}$.

If $\xi_i$ belongs to a smooth block $I_m$, then
\[
|f(\xi_k,v)-f(\xi_i,v)|
\le
L_m|\xi_k-\xi_i|
\le
C b_n(h_n)
\qquad
\text{for every }v\in[0,1].
\]

If $\xi_i$ belongs to a block with identical probability profiles, then
\[
f(\xi_k,v)=f(\xi_i,v)
\qquad
\text{for every }v\in[0,1].
\]

Consequently, in either case,
\[
\max_{k\in\Ncal_i^{(-j)}}
|P_{kj}-P_{ij}|
\le
C b_n(h_n),
\]
where the left-hand side is identically zero on a block with identical probability profiles.
\end{corollary}

\section{Proofs for neighborhood selection and bias control}
\label{app:geometry-proofs}
We first prove the localization theorem using the auxiliary results of
Appendix~\ref{app:geometry-lemmas}, then give the remaining geometric
proofs.
\begin{proof}[Proof of Theorem~\ref{thm:geo-localization}]
On the event $\Ecal_{\mathrm{master}}$ defined in
\eqref{eq:master-event}, Lemma~\ref{lem:distance-approximation}
gives the first conclusion. Lemmas~\ref{lem:containment} and
\ref{lem:metric-to-slice} give block containment and latent-position
localization. Corollary~\ref{cor:all-slice-close} gives the uniform
pointwise bound on the connection profiles. By
Remark~\ref{rem:master-event}, this event has probability at least
$1-Cn^{-c}$.
\end{proof}

\begin{proof}[Proof of Lemma~\ref{lem:population-twohop}]
Fix distinct $i,\ell,j$ and condition on the latent positions. Since
$A_{ii}=A_{\ell\ell}=0$,
\[
M^{(-j)}_{i\ell}
=
\frac{1}{n-1}
\sum_{m\notin\{i,\ell,j\}}
A_{im}A_{m\ell}.
\]
For $m\notin\{i,\ell,j\}$, define
\[
Y_m
:=
A_{im}A_{m\ell}-P_{im}P_{m\ell}.
\]
For distinct $m\ne m'$, the edge pairs
$\{A_{im},A_{m\ell}\}$ and $\{A_{im'},A_{m'\ell}\}$
are disjoint. Hence the $Y_m$ are independent and mean zero
conditionally on $\boldsymbol\xi$, with $|Y_m|\le1$.
For every $t>0$, Hoeffding's inequality gives
\[
\PP\!\left(
\left|M^{(-j)}_{i\ell}-\frac{1}{n-1}
\sum_{m\notin\{i,\ell,j\}}P_{im}P_{m\ell}\right|>t
\,\middle|\,\boldsymbol\xi\right)
\le 2\exp\!\left(-\frac{(n-1)^2t^2}{2(n-3)}\right).
\]
Let $r_n=\sqrt{\log n/n}$ and choose $t=Ar_n$. A union
bound over at most $n^3$ ordered triples yields, for all $n\ge4$,
\[
\PP\!\left(
\max_{i,\ell,j\,\mathrm{distinct}}
\left|M^{(-j)}_{i\ell}-\frac{1}{n-1}
\sum_{m\notin\{i,\ell,j\}}P_{im}P_{m\ell}\right|>Ar_n
\,\middle|\,\boldsymbol\xi\right)
\le 2n^{3-A^2/2}.
\]
This conditional bound is uniform in the latent configuration.

For the latent fluctuation, condition on $(\xi_i,\xi_\ell)$
and put $Z_m=f(\xi_i,\xi_m)f(\xi_m,\xi_\ell)$ for
$m\notin\{i,\ell,j\}$. The $n-3$ variables $Z_m$ are
conditionally independent and take values in $[0,1]$, with mean
\[
\mu_{i\ell}=\int_0^1 f(\xi_i,x)f(x,\xi_\ell)\,dx.
\]
The same conservative Hoeffding bound and a union bound give
\[
\PP\!\left(
\max_{i,\ell,j\,\mathrm{distinct}}
\left|\frac1{n-1}\sum_{m\notin\{i,\ell,j\}}
(Z_m-\mu_{i\ell})\right|>Ar_n\right)
\le 2n^{3-A^2/2}.
\]
The difference between $(n-3)\mu_{i\ell}/(n-1)$ and
$\mu_{i\ell}$ is at most $2/(n-1)$.
Choosing $A$ sufficiently large and absorbing this deterministic
normalization error proves both asserted bounds with failure
probability at most $Cn^{-c}$, uniformly over the admissible triples.
\end{proof}

\subsection{Proofs of lemmas in Subsection~\ref{subsec:empirical-concentration}: two-hop concentration}
\begin{proof}[Proof of Lemma~\ref{lem:twohop-concentration}]
By Lemma~\ref{lem:population-twohop}, uniformly over all admissible
indices,
\[
M^{(-j)}_{i\ell}
=
\frac{1}{n-1}
\sum_{m\notin\{i,\ell,j\}}
P_{im}P_{m\ell}
+
R_{i\ell j},
\qquad
|R_{i\ell j}|
\le
Cr_n,
\]
and similarly for $M^{(-j)}_{k\ell}$. Subtracting the two expansions and
extracting the common summation range gives
\[
\begin{aligned}
M^{(-j)}_{i\ell}-M^{(-j)}_{k\ell}
&=
\frac{1}{n-1}
\sum_{m\notin\{i,k,\ell,j\}}
(P_{im}-P_{km})P_{m\ell}\\
&\quad+
\frac{
P_{ik}P_{k\ell}
-
P_{ki}P_{i\ell}
}{n-1}
+
R_{i\ell j}-R_{k\ell j}.
\end{aligned}
\]
Since $0\le P_{ab}\le1$,
\(
\left|
\frac{
P_{ik}P_{k\ell}
-
P_{ki}P_{i\ell}
}{n-1}
\right|
\le
\frac{2}{n-1}
=
o(r_n).
\)
The two remainder terms are uniformly $O(r_n)$, which proves the claimed
bound after enlarging the constant.
\end{proof}
\begin{proof}[Proof of Lemma~\ref{lem:integral-to-sum}]
By conditioning on the latent positions $\xi_i, \xi_{k'},$ and $\xi_\ell$, the summands
\[
Y_m := (f(\xi_i,\xi_m)-f(\xi_{k'},\xi_m)) f(\xi_m,\xi_\ell),
\qquad m \notin \{i,j,k',\ell\},
\]
are independent and identically distributed, with conditional mean
\[
\EE[Y_m \mid \xi_i,\xi_{k'},\xi_\ell]
=
\int_0^1 (f(\xi_i,x)-f(\xi_{k'},x))f(x,\xi_\ell)\,dx.
\] Since the $n-4$ variables $Y_m$ are conditionally independent and
lie in $[-1,1]$, Hoeffding's inequality gives, for any $t>0$,
\[
\PP\!\left(
\left|\frac1{n-1}\sum_{m\notin\{i,j,k',\ell\}}
(Y_m-\EE[Y_m\mid\xi_i,\xi_{k'},\xi_\ell])\right|>t
\,\middle|\,\xi_i,\xi_{k'},\xi_\ell\right)
\le2\exp\!\left(-\frac{(n-1)^2t^2}{2(n-4)}\right).
\]
Putting $t=Ar_n$ and using at most $n^4$ admissible index
tuples gives
\[
\PP\!\left(
\max_{\mathrm{admissible}\ (i,j,k',\ell)}
\left|\frac1{n-1}\sum_{m\notin\{i,j,k',\ell\}}
(Y_m-\EE[Y_m\mid\xi_i,\xi_{k'},\xi_\ell])\right|
>Ar_n\right)
\le2n^{4-A^2/2}
\]
for all sufficiently large $n$. The difference between the
normalized conditional expectation of the finite sum and the graphon
integral is at most $3/(n-1)$, because the integral is bounded in
absolute value by one. Choosing $A$ sufficiently large and absorbing
this deterministic error proves the stated polynomial failure bound.
\end{proof}
\begin{proof}[Proof of Lemma~\ref{lem:max-to-sup}]
Fix $i,k'$ and let $g_{ik'}$ be defined by \eqref{eq:twohop-difference-function}.

Because $f$ is symmetric, the blockwise regularity in the first argument
also controls the second argument. If $I_m$ is smooth and
$w,w'\in I_m$, then for every $x\in[0,1]$,
\[
|f(x,w)-f(x,w')|
=
|f(w,x)-f(w',x)|
\le
L_m|w-w'|.
\]
If $I_m$ is a block with identical probability profiles, symmetry and Assumption~\ref{as:hybrid}(2) instead give
\[
f(x,w)=f(x,w')
\qquad
\text{for all }x\in[0,1],\quad w,w'\in I_m.
\]
Consequently, writing
\(
L_*:=
\max\Bigl(
\{L_m:I_m\text{ is smooth}\}\cup\{0\}
\Bigr),
\)
we have for $w,w'$ belonging to the same block,
\begin{align*}
|g_{ik'}(w)-g_{ik'}(w')|
&\le
\int_0^1
|f(\xi_i,x)-f(\xi_{k'},x)|
\,|f(x,w)-f(x,w')|
\,dx \\
&\le
L_*|w-w'|.
\end{align*}
Thus $g_{ik'}$ is uniformly Lipschitz on each block $I_m$.
Let
\[
\delta_n:=C_{\mathrm{mesh}}\frac{\log n}{n},
\]
where $C_{\mathrm{mesh}}>0$ is a sufficiently large constant. Since the number of blocks is
fixed and every block has fixed positive length, for all sufficiently large
$n$ we may partition each block $I_m$ into intervals
$\{J_{m,r}^{(n)}\}_r$ whose lengths lie between $\delta_n$ and
$2\delta_n$.

For each such interval $J$, the number
\(
N_J:=\sum_{\ell=1}^n\mathbbm{1}\{\xi_\ell\in J\}
\)
is binomial with mean at least $C_{\mathrm{mesh}}\log n$.
A multiplicative Chernoff bound therefore gives
\[
\PP(N_J<4)
\le
\exp(-cC_{\mathrm{mesh}}\log n)
\]
for some universal constant $c>0$. The total number of intervals is
$O(n/\log n)$, so, after choosing $C_{\mathrm{mesh}}$ sufficiently large
and taking a union bound,
\[
\PP\!\left(
\min_J N_J\ge4
\right)
\ge
1-Cn^{-c'}.
\]

On this event, deleting any three indices still leaves at least one
admissible sampled latent position in every interval. Hence, uniformly over
all admissible triples $(i,k',j)$, every $w\in I_m$ has some
$\ell\notin\{i,k',j\}$ with $\xi_\ell\in I_m$ such that
\(
|w-\xi_\ell|
\le
2\delta_n.
\)
By the blockwise Lipschitz bound established above,
\(
\bigl|
g_{ik'}(w)-g_{ik'}(\xi_\ell)
\bigr|
\le
2L_*\delta_n.
\)
Therefore,
\(
\sup_{w\in I_m}|g_{ik'}(w)|
\le
\max_{\substack{\ell:\,\xi_\ell\in I_m\\
\ell\notin\{i,k',j\}}}
|g_{ik'}(\xi_\ell)|
+
2L_*\delta_n.
\)
The reverse inequality holds trivially because every sampled
$\xi_\ell\in I_m$ is an admissible point in the supremum over the block.
Consequently,
\[
\left|
\max_{\substack{\ell:\,\xi_\ell\in I_m\\
\ell\notin\{i,k',j\}}}
|g_{ik'}(\xi_\ell)|
-
\sup_{w\in I_m}|g_{ik'}(w)|
\right|
\le
C\frac{\log n}{n}.
\]
Since the number of blocks is fixed, the result holds uniformly over all
blocks and all admissible triples $(i,k',j)$.
\end{proof}
\begin{proof}[Proof of Corollary~\ref{cor:empirical-to-integral}]
Fix distinct $i,k',\ell\in V\setminus\{j\}$ and define
\[
Q_{ik'\ell j}
:=
\frac{1}{n-1}
\sum_{m\notin\{i,k',\ell,j\}}
(P_{im}-P_{k'm})P_{m\ell}.
\]
Then
\[
\begin{aligned}
&\left|
\bigl(M^{(-j)}_{i\ell}-M^{(-j)}_{k'\ell}\bigr)
-
\int_0^1
\{f(\xi_i,x)-f(\xi_{k'},x)\}f(x,\xi_\ell)\,dx
\right|\\
&\qquad\le
\left|
\bigl(M^{(-j)}_{i\ell}-M^{(-j)}_{k'\ell}\bigr)
-
Q_{ik'\ell j}
\right|\\
&\qquad\quad+
\left|
Q_{ik'\ell j}
-
\int_0^1
\{f(\xi_i,x)-f(\xi_{k'},x)\}f(x,\xi_\ell)\,dx
\right|.
\end{aligned}
\]
The first term is bounded uniformly by
Lemma~\ref{lem:twohop-concentration}, and the second is bounded uniformly
by Lemma~\ref{lem:integral-to-sum}. Their intersection therefore gives
\[
\left|
\bigl(M^{(-j)}_{i\ell}-M^{(-j)}_{k'\ell}\bigr)
-
\int_0^1
\{f(\xi_i,x)-f(\xi_{k'},x)\}f(x,\xi_\ell)\,dx
\right|
\le
Cr_n
\]
uniformly over all admissible indices.
\end{proof}

\begin{proof}[Proof of Lemma~\ref{lem:distance-approximation}]
Fix distinct $i,k,j$ and define the sampled population discrepancy
\[
D^{\mathrm{disc}}_{ikj}
:=
\max_{\ell\notin\{i,k,j\}}
|g_{ik}(\xi_\ell)|,
\]
where $g_{ik}$ is defined in \eqref{eq:twohop-difference-function}.
On $\Ecal_{\mathrm{conc}}\cap\Ecal_{\mathrm{int}}$, Corollary~\ref{cor:empirical-to-integral}
gives, uniformly in $\ell$,
\[
\left|
\bigl(M^{(-j)}_{i\ell}-M^{(-j)}_{k\ell}\bigr)-g_{ik}(\xi_\ell)
\right|
\le Cr_n.
\]
Using the elementary inequality
\[
\left|
\max_{\ell}|a_\ell|-
\max_{\ell}|b_\ell|
\right|
\le
\max_{\ell}|a_\ell-b_\ell|,
\]
we obtain
\(
\left|
 d^{(-j)}(i,k)-D^{\mathrm{disc}}_{ikj}
\right|
\le Cr_n.
\)
On $\Ecal_{\mathrm{fill}}$, Lemma~\ref{lem:max-to-sup} and the blockwise
representation \eqref{eq:D2-blockwise} imply
\[
\left|
D^{\mathrm{disc}}_{ikj}-D_2(\xi_i,\xi_k)
\right|
\le
C\frac{\log n}{n}.
\]
Because $(\log n)/n\le r_n$ for all sufficiently large $n$, the triangle
inequality yields \eqref{eq:distance-approximation} on
$\Ecal_0=\Ecal_{\mathrm{conc}}\cap\Ecal_{\mathrm{int}}\cap\Ecal_{\mathrm{fill}}$.
The bounds are uniform over all admissible triples, completing the proof.
\end{proof}

\subsection{Proofs of Subsection~\ref{subsec:neighborhood-containment}: neighborhood containment}

\begin{proof}[Proof of Lemma~\ref{lem:candidates}]
Put $\ell_*:=\min_m|I_m|>0$. As $h_n=o(n)$, for sufficiently
large $n$, $2t_n<\ell_*$. Fix $i\ne j$ and condition on
$\xi_i=u\in I_m$. If $I_m$ is smooth, the candidate window
$J_m(u):=I_m\cap[u-t_n,u+t_n]$ has Lebesgue measure at least $t_n$.
Indeed, either $[u-t_n,u]$ or $[u,u+t_n]$ lies in the
closure of $I_m$ (depending on the distance to the left boundary).
Possible inclusion of an endpoint has no effect on Lebesgue measure.
Conditional on $\xi_i=u$,
\[
|S_{ij}|\sim\operatorname{Binomial}(n-2,p_m(u)),
\qquad p_m(u)=\lambda(J_m(u))\ge t_n.
\]
For any fixed $q\in(1,a)$, its conditional mean is at least
$qh_n$ for sufficiently large $n$. The Chernoff inequality gives
\[
\PP(|S_{ij}|<h_n\mid\xi_i=u)
\le\exp\left[-\frac{(q-1)^2}{2q}h_n\right].
\]
This estimate is uniform in $u$, including positions arbitrarily
close to a block endpoint. If $I_m$ has identical profiles, the
candidate count is instead binomial with parameters $(n-2,|I_m|)$.
Since $|I_m|\ge\ell_*$ and $h_n=o(n)$, the lower-tail
probability is at most $e^{-cn}$. A union bound over at most
$n(n-1)$ ordered pairs proves the first assertion when
$C_0$ is chosen sufficiently large.

For any candidate $k$, both latent positions lie in the same block.
If this block is smooth, the within-block Lipschitz property and
$0\le f\le1$ give
\[
D_2(\xi_i,\xi_k)
\le L_m|\xi_i-\xi_k|\le L_mt_n.
\]
If it has identical profiles, the discrepancy is zero.
Lemma~\ref{lem:distance-approximation} then yields
\[
d^{(-j)}(i,k)\le D_2(\xi_i,\xi_k)+Cr_n
\le Cb_n(h_n)
\]
on $\Ecal_0$, simultaneously for all admissible candidates.
\end{proof}

\begin{proof}[Proof of Lemma~\ref{lem:slice-close}]
Because $k\in\Ncal_i^{(-j)}$, it is one of the $h_n$ nearest neighbors of
$i$ under the empirical LOO distance. On $\Ecal_{\mathrm{master}}$,
Lemma~\ref{lem:candidates} provides at least $h_n$ candidate vertices whose
empirical distances from $i$ are bounded by $C b_n(h_n)$. Therefore
\[
d^{(-j)}(i,k)
\le
d^{(-j)}_{i,(h_n)}
\le
C b_n(h_n).
\]
By Lemma~\ref{lem:distance-approximation},
\[
D_2(\xi_i,\xi_k)
\le
d^{(-j)}(i,k)+Cr_n.
\]
Since $r_n\le b_n(h_n)$, the preceding nearest-neighbor bound gives
\(
D_2(\xi_i,\xi_k)
\le
C b_n(h_n),
\)
which proves the claim.
\end{proof}

\subsection{Proofs of Subsection~\ref{subsec:bias-inversion}: latent geometry}

\begin{proof}[Proof of Lemma~\ref{lem:injectivity}]
This is an immediate rearrangement of
Assumption~\ref{as:nondeg-block}(1).
\end{proof}

\begin{proof}[Proof of Lemma~\ref{lem:metric-to-slice}]
Fix an ordered pair $(i,j) \in \Ecal_n$ with $i$ in block $I_m$. By Lemma~\ref{lem:containment}, all selected neighbors $k \in \Ncal_i^{(-j)}$ belong to the same block $I_m$. We consider two cases:

For a smooth block, Lemma~\ref{lem:twohop-to-slice} and the lower
bi-Lipschitz bound in Assumption~\ref{as:hybrid}(1) give
\[
c_m|\xi_i-\xi_k|
\le
\|f(\xi_i,\cdot)-f(\xi_k,\cdot)\|_{L^2}
\le
C_0 b_n(h_n).
\]
Hence
\(
|\xi_i-\xi_k|
\le
\frac{C_0}{c_m}b_n(h_n).
\)
The pointwise Lipschitz property then gives, uniformly in $j$,
\[
|P_{kj}-P_{ij}|
=
|f(\xi_k,\xi_j)-f(\xi_i,\xi_j)|
\le
L_m|\xi_k-\xi_i|
\le
C b_n(h_n).
\]
Averaging over $k\in\Ncal_i^{(-j)}$ proves the stated smooth-block
bias bound. For a block with identical probability profiles, Lemma~\ref{lem:containment} places all
selected neighbors in the same block with identical probability profiles, and
Assumption~\ref{as:hybrid}(2) gives $P_{kj}=P_{ij}$ for every selected
$k$, proving the stated zero-bias conclusion.
\end{proof}

\begin{proof}[Proof of Corollary~\ref{cor:all-slice-close}]
Fix $(i,j)\in\Ecal_n$ and
$k\in\Ncal_i^{(-j)}$ on $\Ecal_{\mathrm{master}}$. By
Lemma~\ref{lem:metric-to-slice}, $k$ belongs to the same block as $i$.
If this is a smooth block $I_m$, the lemma also gives
$|\xi_k-\xi_i|\le Cb_n(h_n)$. The pointwise Lipschitz condition in
Assumption~\ref{as:hybrid}(1), applied at any $v\in[0,1]$, yields
\[
|f(\xi_k,v)-f(\xi_i,v)|
\le L_m|\xi_k-\xi_i|
\le L_m Cb_n(h_n).
\]
The partition is fixed and finite, and each $L_m$ is fixed, so the
maximum of these constants over the smooth blocks can be absorbed into
$C$, uniformly in $(i,j)$, $k$, and $v$. If $i$ belongs to an
block with identical probability profiles, Assumption~\ref{as:hybrid}(2) gives
$f(\xi_k,v)=f(\xi_i,v)$ for every $v\in[0,1]$. Taking $v=\xi_j$ proves
the asserted bound on $|P_{kj}-P_{ij}|$ and hence the corollary.
\end{proof}

\begin{proof}[Proof of Lemma~\ref{lem:containment}]
Let $k'\in V\setminus\{i,j\}$ satisfy $\xi_{k'}\notin I_m$. By the across-block two-hop separation in
Assumption~\ref{as:nondeg-block}(2),
\(
D_2(\xi_i,\xi_{k'})\ge\eta_2.
\)
On $\Ecal_{\mathrm{master}}\subseteq\Ecal_0$,
Lemma~\ref{lem:distance-approximation} therefore implies
\[
d^{(-j)}(i,k')
\ge
D_2(\xi_i,\xi_{k'})-C r_n
\ge
\eta_2-C r_n.
\]
Conversely, every candidate $k\in S_{ij}$ satisfies
\(
d^{(-j)}(i,k)
\le
C b_n(h_n)
\)
by Lemma~\ref{lem:candidates}. Since $h_n/n\to0$ and $r_n\to0$,
\[
b_n(h_n)\to0.
\]
Hence, for all sufficiently large $n$,
\[
C b_n(h_n)
<
\eta_2-C r_n.
\]
Thus every candidate vertex is empirically closer to $i$ than every
vertex whose latent position lies outside $I_m$. Since $|S_{ij}|\ge h_n$,
the $h_n$ selected nearest neighbors must all have latent positions in
$I_m$, which proves the stated containment.
\end{proof}

\begin{proof}[Proof of Lemma~\ref{lem:twohop-to-slice}]
By Lemma~\ref{lem:containment}, every selected neighbor
$k\in\Ncal_i^{(-j)}$ belongs to the same smooth block $I_m$ as $i$.
Lemma~\ref{lem:slice-close} gives
\[
\sup_{w\in I_m}
\left|
\int_0^1
\{f(\xi_i,x)-f(\xi_k,x)\}f(x,w)\,dx
\right|
\le
C_1 b_n(h_n).
\]
Applying Lemma~\ref{lem:injectivity},
\[
\|f(\xi_i,\cdot)-f(\xi_k,\cdot)\|_{L^2}
\le
\kappa_m^{-1}C_1 b_n(h_n),
\]
which proves the claim after absorbing $\kappa_m^{-1}$ into the constant.
\end{proof}

\clearpage
\section{Additional simulation experiments}
\label{app:stress-test-simulations}
We give additional simulation results for four graphons that do not all
satisfy the assumptions of the main theoretical results. The oscillatory
model is examined across four sample sizes; the remaining three models
are evaluated at a fixed network size. These experiments assess the
finite-sample behavior of the estimator and confidence intervals, but
do not establish asymptotic validity for these graphon classes.

\paragraph{Oscillatory graphon.}

To examine behavior when quantitative two-hop identifiability has not been
verified, we consider
\begin{equation}
\label{eq:sim-wiggly}
f_{\mathrm{wig}}(u,v)
=
0.5+0.25\sin(4\pi uv).
\end{equation}
This graphon satisfies
\[
f_{\mathrm{wig}}([0,1]^2)\subseteq[0.25,0.75],
\]
so it remains bounded away from zero and one. However, we have not established
that it satisfies Assumption~\ref{as:nondeg-block}. We therefore treat this
experiment strictly as a stress test and do not use it as evidence for the
main theorems. For the oscillatory graphon, we use the same target-pair construction, network sizes, graph-replication counts, and inference bandwidth as for each principal graphon in Section~\ref{subsec:sim-design}.

Table~\ref{tab:wiggly-stress} reports estimated feasible Gaussian coverage,
interval width, the Monte Carlo mean empirical-to-oracle variance ratio,
and the mean standardized conditional bias
$|B_{ij}^{(-j)}|/\sqrt{V_{ij}}$.

\input{appendix_tables/wiggly.tex}

The replication-level Monte Carlo standard errors of the LOO Gaussian
coverage estimates range from $0.010$ to $0.014$.
The reported mean empirical-to-oracle variance ratios range from
$1.003$ to $1.034$. The
mean standardized conditional bias decreases from $0.472$ at $n=200$ to
$0.213$ at $n=2000$. Estimated LOO Gaussian coverage is $0.878$, $0.910$,
$0.955$, and $0.950$ at $n=200,500,1000,2000$, respectively. The substantial
finite-sample shortfall at the two smaller sample sizes is accompanied
by larger standardized bias; the reported mean variance ratios alone do
not establish concentration of individual ratios. The
near-nominal estimates at $n=1000$ and $n=2000$ are descriptive and are not
interpreted as establishing asymptotic validity outside the verified theorem
class.

\paragraph{Other stress-test graphons.}

The remaining three models are
\begin{align}
f_{\mathrm{sin}}(u,v)
&=
0.5+0.3\sin(\pi u)\sin(\pi v),
\label{eq:stress-sinusoidal}\\
f_{\mathrm{add}}(u,v)
&=
\frac{u+v}{2},
\label{eq:stress-additive}\\
f_{\mathrm{spike}}(u,v)
&=
0.2+0.8\,
\mathbbm{1}\{|u-0.5|<0.1,\ |v-0.5|<0.1\}.
\label{eq:stress-spiky}
\end{align}

The sinusoidal model satisfies
$f_{\mathrm{sin}}(u,\cdot)=f_{\mathrm{sin}}(1-u,\cdot)$, so mirrored latent
positions have identical profiles. Placing such positions in the same smooth
block is incompatible with the lower $L^2$ slice-distance condition, whereas
placing them in different interval blocks is incompatible with positive
across-block two-hop separation. The additive linear model is not uniformly
bounded away from zero and one because it attains both endpoints. The spiky
model also violates nondegeneracy by attaining one. In addition, under the
natural three-interval partition induced by its central square, the two outer
regions have identical profiles and therefore cannot satisfy positive
two-hop separation when treated as distinct blocks. None of these designs is
used to verify the main inferential theorems.

We use $n=500$, $120$ independently generated graphs, and the same four target
coordinate pairs, giving $480$ target-edge evaluations per model. For point
estimation, we use $h_n=\operatorname{round}(n^{2/3})=63$; for Gaussian and
empirical Bernstein inference, we use
$h_n=\operatorname{round}(n^{3/5})=42$.

\input{appendix_tables/additional_stress.tex}

The replication-level Monte Carlo standard error for each LOO Gaussian coverage
estimate in Table~\ref{tab:additional-stress-tests} is approximately $0.011$.
The corresponding estimated empirical Bernstein coverage for the conditional
neighborhood mean is $1.000$ in all three cases, with average truncated widths
$0.803$, $0.857$, and $0.622$ for the sinusoidal, additive linear, and spiky
designs, respectively. The reported mean empirical-to-oracle variance ratios are close to
one in the first two designs and equal $0.978$ in the spiky design.
These averages do not describe the dispersion of the individual ratios.

\bibliographystyle{plainnat}

\bibliography{Reference_v2}

\end{document}

%% file: generated_results/regenerated_20260831/tables/gaussian.tex
\begin{table}[tb]
\centering
\caption{Feasible Gaussian inference under
$h_n=\operatorname{round}(n^{3/5})$. Widths are truncated to $[0,1]$.
The column $\widehat V/V$ reports the Monte Carlo mean of $\widehat V_{ij}^{\mathrm{emp}}/V_{ij}$.
The ZLZ interval is a descriptive plug-in Gaussian comparator.}
\label{tab:gaussian-inference}
\begin{tabular}{lrrrrrr}
\toprule
Model & $n$ & $h_n$ & LOO cov. & ZLZ cov. & LOO width & $\widehat V/V$ \\
\midrule
Kink & 200 & 24 & 0.937 & 0.923 & 0.397 & 1.007 \\
Kink & 500 & 42 & 0.946 & 0.935 & 0.301 & 1.004 \\
Kink & 1000 & 63 & 0.932 & 0.930 & 0.245 & 1.002 \\
Kink & 2000 & 96 & 0.954 & 0.958 & 0.199 & 1.001 \\
\addlinespace
Linear & 200 & 24 & 0.945 & 0.925 & 0.396 & 1.010 \\
Linear & 500 & 42 & 0.935 & 0.929 & 0.299 & 1.002 \\
Linear & 1000 & 63 & 0.945 & 0.940 & 0.245 & 1.001 \\
Linear & 2000 & 96 & 0.958 & 0.975 & 0.199 & 1.004 \\
\addlinespace
SBM & 200 & 24 & 0.957 & 0.948 & 0.366 & 1.008 \\
SBM & 500 & 42 & 0.929 & 0.898 & 0.276 & 0.999 \\
SBM & 1000 & 63 & 0.938 & 0.932 & 0.226 & 0.998 \\
SBM & 2000 & 96 & 0.975 & 0.971 & 0.183 & 1.001 \\
\bottomrule
\end{tabular}
\end{table}

%% file: generated_results/regenerated_20260831/tables/coverage_diagnostics_v46.tex
\begin{table}[tb]
\centering
\small
\caption{Graph-replication-level diagnostics for feasible LOO Gaussian confidence intervals at nominal coverage $0.95$. Here $\widehat p-0.95$ is signed coverage error, $\mathrm{MCSE}$ is the standard deviation of the four-target coverage averages across independent graphs divided by the square root of the number of graph replications, and $\overline{|B|/\sqrt V}$ is mean absolute standardized conditional smoothing bias. All entries use the archived raw simulation output.}
\label{tab:gaussian-coverage-diagnostics}
\begin{tabular}{lrrrrrr}
\toprule
Model & $n$ & Graphs & $\widehat p$ & $\widehat p-0.95$ & MCSE & $\overline{|B|/\sqrt V}$ \\
\midrule
Kink & 200 & 150 & 0.937 & -0.013 & 0.011 & 0.190 \\
Kink & 500 & 120 & 0.946 & -0.004 & 0.011 & 0.195 \\
Kink & 1000 & 100 & 0.932 & -0.017 & 0.012 & 0.166 \\
Kink & 2000 & 60 & 0.954 & +0.004 & 0.013 & 0.155 \\
\addlinespace
Linear & 200 & 150 & 0.945 & -0.005 & 0.009 & 0.256 \\
Linear & 500 & 120 & 0.935 & -0.015 & 0.010 & 0.242 \\
Linear & 1000 & 100 & 0.945 & -0.005 & 0.012 & 0.214 \\
Linear & 2000 & 60 & 0.958 & +0.008 & 0.012 & 0.182 \\
\addlinespace
SBM & 200 & 150 & 0.957 & +0.007 & 0.008 & 0.000 \\
SBM & 500 & 120 & 0.929 & -0.021 & 0.012 & 0.000 \\
SBM & 1000 & 100 & 0.938 & -0.012 & 0.012 & 0.000 \\
SBM & 2000 & 60 & 0.975 & +0.025 & 0.010 & 0.000 \\
\bottomrule
\end{tabular}
\end{table}

%% file: generated_results/regenerated_20260831/tables/bernstein.tex
\begin{table}[tb]
\centering
\caption{Empirical Bernstein performance for the conditional
neighborhood mean under $h_n=\operatorname{round}(n^{3/5})$. ``Full'' is
the fraction of truncated intervals equal to $[0,1]$.}
\label{tab:bernstein-inference}
\begin{tabular}{lrrrrr}
\toprule
Model & $n$ & $h_n$ & EB cov. & EB width & Full \\
\midrule
Kink & 200 & 24 & 1.000 & 0.997 & 0.960 \\
Kink & 500 & 42 & 1.000 & 0.899 & 0.000 \\
Kink & 1000 & 63 & 1.000 & 0.699 & 0.000 \\
Kink & 2000 & 96 & 1.000 & 0.515 & 0.000 \\
\addlinespace
Linear & 200 & 24 & 1.000 & 0.997 & 0.958 \\
Linear & 500 & 42 & 1.000 & 0.891 & 0.000 \\
Linear & 1000 & 63 & 1.000 & 0.697 & 0.000 \\
Linear & 2000 & 96 & 1.000 & 0.515 & 0.000 \\
\addlinespace
SBM & 200 & 24 & 1.000 & 0.962 & 0.615 \\
SBM & 500 & 42 & 1.000 & 0.758 & 0.000 \\
SBM & 1000 & 63 & 1.000 & 0.627 & 0.000 \\
SBM & 2000 & 96 & 1.000 & 0.491 & 0.000 \\
\bottomrule
\end{tabular}
\end{table}

%% file: generated_results/regenerated_20260831/tables/rate.tex
\begin{table}[tb]
\centering
\caption{Rate experiment at
$h_n=\operatorname{round}(n^{2/3})$. ``Scaled'' denotes
$n^{2/3}$ times the LOO MSE.}
\label{tab:rate-simulation}
\begin{tabular}{lrrrrr}
\toprule
Model & $n$ & $h_n$ & LOO MSE & ZLZ MSE & Scaled \\
\midrule
Kink & 200 & 34 & 0.00779 & 0.00843 & 0.266 \\
Kink & 500 & 63 & 0.00438 & 0.00424 & 0.276 \\
Kink & 1000 & 100 & 0.00248 & 0.00268 & 0.248 \\
Kink & 2000 & 159 & 0.00169 & 0.00168 & 0.269 \\
\addlinespace
Linear & 200 & 34 & 0.00792 & 0.00803 & 0.271 \\
Linear & 500 & 63 & 0.00456 & 0.00474 & 0.288 \\
Linear & 1000 & 100 & 0.00286 & 0.00300 & 0.286 \\
Linear & 2000 & 159 & 0.00171 & 0.00170 & 0.271 \\
\addlinespace
SBM & 200 & 34 & 0.00661 & 0.00646 & 0.226 \\
SBM & 500 & 63 & 0.00337 & 0.00373 & 0.212 \\
SBM & 1000 & 100 & 0.00218 & 0.00228 & 0.218 \\
SBM & 2000 & 159 & 0.00125 & 0.00140 & 0.198 \\
\bottomrule
\end{tabular}
\end{table}

%% file: generated_results/regenerated_20260831/tables/bandwidth_sweep.tex
\begin{table}[tb]
\centering
\caption{Independent neighborhood-size sweep for the kink
graphon at $n=1000$. The final column is the Monte Carlo mean of
$|B_{ij}^{(-j)}|/\sqrt{V_{ij}}$.}
\label{tab:bandwidth-sweep}
\begin{tabular}{rrrr}
\toprule
$c$ & $h_n$ & LOO MSE & Mean standardized bias \\
\midrule
0.5 & 50 & 0.00519 & 0.166 \\
1.0 & 100 & 0.00278 & 0.231 \\
2.0 & 200 & 0.00149 & 0.320 \\
4.0 & 400 & 0.00075 & 0.436 \\
\bottomrule
\end{tabular}
\end{table}

%% file: appendix_tables/wiggly.tex
\begin{table}[H]
\centering
\caption{Oscillatory stress test under the inference bandwidth.
The design is not used to validate the main theorem because quantitative
two-hop identifiability has not been verified. The column $\widehat V/V$ is the Monte Carlo mean of the empirical-to-oracle variance ratio.}
\label{tab:wiggly-stress}
\begin{tabular}{rrrrrrr}
\toprule
$n$ & $h_n$ & LOO cov. & ZLZ cov. & LOO width & $\widehat V/V$ & Std. bias \\
\midrule
200 & 24 & 0.878 & 0.885 & 0.373 & 1.034 & 0.472 \\
500 & 42 & 0.910 & 0.890 & 0.277 & 1.015 & 0.308 \\
1000 & 63 & 0.955 & 0.955 & 0.226 & 1.022 & 0.251 \\
2000 & 96 & 0.950 & 0.946 & 0.181 & 1.003 & 0.213 \\
\bottomrule
\end{tabular}
\end{table}

%% file: appendix_tables/additional_stress.tex
\begin{table}[H]
\centering
\caption{Additional stress-test results at $n=500$.
Point-estimation MSEs use $h_n=63$ and Gaussian inference uses $h_n=42$. The column $\widehat V/V$ reports a Monte Carlo mean.}
\label{tab:additional-stress-tests}
\small
\setlength{\tabcolsep}{4pt}
\begin{tabular}{lrrrrrr}
\toprule
Model & LOO MSE & ZLZ MSE & LOO cov. & ZLZ cov. & LOO width & $\widehat V/V$ \\
\midrule
Sinusoidal & 0.00377 & 0.00379 & 0.948 & 0.935 & 0.286 & 1.001 \\
Additive linear & 0.00463 & 0.00470 & 0.935 & 0.925 & 0.294 & 0.996 \\
Spiky & 0.00249 & 0.00252 & 0.921 & 0.925 & 0.237 & 0.978 \\
\bottomrule
\end{tabular}
\end{table}